\documentclass[journal,onecolumn]{IEEEtran}

\usepackage{xcolor}

\usepackage{algorithm}
\usepackage{algpseudocode}
\usepackage{float}
\usepackage{graphics}
\usepackage{epsf} 
\usepackage[bookmarks=false]{hyperref} 

\usepackage{url}

\usepackage{cite}

\usepackage{amssymb,hhline,enumerate,dsfont}
\usepackage{amsmath}
\usepackage{array}
\usepackage{psfrag}
\usepackage{epsfig}
\usepackage{amsthm}
\usepackage{amsmath,amssymb,times}
\usepackage{multirow,epsfig,cite}

\usepackage[utf8]{inputenc}
\usepackage{multirow}
\usepackage{mathtools}

\def\block(#1,#2)#3{\multicolumn{#2}{c}{\multirow{#1}{*}{$ #3 $}}}

\newtheorem{thm}{Theorem}
\newtheorem{rmrk}{Remark}
\newtheorem{lem}{Lemma}

\newtheorem{cor}{Corollary}
\newtheorem{defn}{Definition}

\theoremstyle{definition}
\newtheorem{examplex}{Example}
\newenvironment{example}
  {\pushQED{\qed}\examplex}
  {\popQED\endexamplex}

\begin{document}
\title{Bandwidth Adaptive \& Error Resilient MBR Exact Repair Regenerating Codes}

\author{\IEEEauthorblockN{Kaveh Mahdaviani\IEEEauthorrefmark{1}, Ashish Khisti\IEEEauthorrefmark{1}, and Soheil Mohajer\IEEEauthorrefmark{2} }\\
\IEEEauthorblockA{\IEEEauthorrefmark{1}ECE Dept., University of Toronto, Toronto, ON M5S3G4, Canada\\
\IEEEauthorrefmark{2}ECE Dept., University of Minnesota, Minneapolis, MN 55455, USA \\              
Email: \{kaveh, akhisti\}@comm.utoronto.ca,~soheil@umn.edu}}

\maketitle

\begin{abstract}
Regenerating codes are efficient methods for distributed storage in storage networks, where node failures are common. They guarantee low cost data reconstruction and repair through accessing only a predefined number of arbitrarily chosen storage nodes in the network. In this work we consider two simultaneous extensions to the original regenerating codes framework introduced in \cite{Regenerating}; \emph{i}) both data reconstruction and repair are resilient to the presence of a certain number of erroneous nodes in the network and \emph{ii}) the number of helper nodes in every repair is not fixed, but is a flexible parameter that can be selected during the run-time. We study the fundamental limits of required total repair bandwidth and provide an upper bound for the storage capacity of these codes under these assumptions. We then focus on the minimum repair bandwidth (MBR) case and derive the exact storage capacity by presenting explicit coding schemes with exact repair, which achieve the upper bound of the storage capacity in the considered setup. To this end, we first provide a more natural extension of the well-known Product Matrix (PM) MBR codes \cite{PM_Codes}, modified to provide flexibility in the choice of number of helpers in each repair, and simultaneously be robust to erroneous nodes in the network. This is achieved by proving the non-singularity of family of matrices in large enough finite fields. We next provide another extension of the PM codes, based on novel repair schemes which enable flexibility in the number of helpers and robustness against erroneous nodes without any extra cost in field size compared to the original PM codes.\footnote{Parts of these results has been presented in IEEE International Symposium on Information Theory (ISIT'16) \cite{BAER_ISIT16}}

\end{abstract}

\IEEEpeerreviewmaketitle

\section{Introduction}

Many distributed storage systems (DSS) are at work these days to store huge digital contents, accessed by many users from different locations. These DSSs typically consist of many storage nodes each with limited storage capacity and service rate. These nodes collaborate to store and serve the data to the users, while from time to time a storage node or part of its data become inaccessible for various reasons. This event is referred to as a \emph{"failure"}. As the number of storage nodes increase, failures have become a norm in nowadays large scale DSSs.

Efficiency and robustness against data loss and data corruption is a key feature of any storage system. As a result, a reliable DSS is required to be able to guarantee the recovery of the stored data, referred to as \emph{data reconstruction}, even if a certain number of nodes are unavailable. This requires a mechanism to store some redundancy in the DSS, which results in some storage overhead. Moreover, to maintain the level of redundancy the DSS needs to be capable of replacing any missing node by a new one which contains an equivalent data. This procedure is referred to as \emph{repair}. In order to perform the repair, it is required to download some data from the remaining nodes in the DSS, which is referred to as the \emph{repair bandwidth}. Therefore, efficiency of a reliable DSS could be measured in terms of its storage overhead and required repair bandwidth.

Among various methods of performing data reconstruction and repair, \cite{Regenerating_Conf, Regenerating} suggested a framework, named regenerating codes, for storage in a DSS consisting $n$ storage nodes. In this framework, all the symbols in the storage network are considered to be elements of a Galois field of appropriate size $q$, denoted by $\mathbb{F}_q$. Each storage node is supposed to have \emph{per node storage capacity} equal to $\alpha$ symbols. This framework suggests optimizing both storage overhead and repair bandwidth, while guarantees capability of data reconstruction using any arbitrary subset of $k$ nodes, and repair using any arbitrary subset of $d$ nodes, referred to as \emph{helpers}, each providing $\beta$ \emph{per node repair bandwidth}. The aggregate amount of repair data is referred to as the \emph{total repair bandwidth} as is denoted by $\gamma = \beta d$. 

A regenerating code encodes a message consisting of source symbols from $\mathbb{F}_q$ into $n \alpha$ coded symbols to be stored on the $n$ storage nodes in the system, such that data reconstruction and repair processes are possible as described. The maximum size of the source data that could be stored in such a DSS is named the \emph{total storage capacity} of the coding scheme, which we denote by $F$. The aim of the code designer is to design a coding scheme which provides the largest possible total storage capacity $F$, for the set of given parameters $(n,$ $k,$ $d,$ $\alpha,$ $\beta)$. It has been shown in \cite{Regenerating_Conf,Regenerating} that $F$ could be upper bounded for any regenerating coding scheme by
\begin{align}\label{eq_DimakisBound}
F \leq \sum_{i=0}^{\min{(k,d)}-1}{\min\left\lbrace \alpha,(d-i)\frac{\gamma}{d} \right\rbrace}.
\end{align}

Based on (\ref{eq_DimakisBound}), It is shown that in an optimally efficient regenerating code with a given total storage capacity, reducing the storage overhead requires increasing the repair bandwidth and vice-versa \cite{Regenerating}. This result is known as the storage-bandwidth tradeoff. In particular, for given $k$, and $d$, the set of pairs of per node storage capacity and per node repair bandwidth, $(\alpha, \beta)$ achieving the equality in (\ref{eq_DimakisBound}) are the optimal pairs. One could see that there is a trade-off between the two parameters $\alpha$, and $\beta$. Among the possible optimal choices for $\alpha$ and $\beta$ for a given value of $F$, one extreme point referred to as the \emph{minimum storage regenerating} (MSR) codes, could be achieved by first minimizing the per node storage capacity $\alpha$ as 
\begin{align}
\alpha_{\text{MSR}} = \frac{F}{k}, \nonumber
\end{align}
for which the optimal value of $\beta$ will be 
\begin{align}
\beta_{\text{MSR}} = \frac{F}{k(d-k+1)}. \nonumber
\end{align}

The other extreme point is obtained by first minimizing the per node repair bandwidth $\beta$, as $\beta_{\text{MBR}}=\alpha/d$, and hence referred to as the \emph{minimum bandwidth regenerating} (MBR) codes, which results in 
\begin{align}
\alpha_{\text{MBR}} = \frac{2F d}{k(2 d-k+1)}, \nonumber
\end{align}
and hence 
\begin{align}
\beta_{\text{MBR}} = \frac{\alpha}{d} = \frac{2 F}{k(2 d-k+1)}. \nonumber
\end{align}

If any replacement node stores exactly the same data as of the corresponding failed node the regenerating code is called an \emph{exact repair} code, and otherwise it is called a \emph{functional repair} code. The existence of functional repair regenerating codes for any point in the trade-off described by (\ref{eq_DimakisBound}), follows from the equivalence of functional repair regenerating code design problem to that of network coding for single source multicast \cite{Deterministic_Regenerating}. However, exact repair regenerating codes are much more appealing from practical point of view e.g., for the fact that they can be tuned as systematic codes. For both MBR and MSR points exact repair codes exists for different parameters \cite{Exact_R_by_T_1, Permutation, PM_Codes, Exact_AI_Asymptotic, MBR_replication, Coupled_Layer, Zigzag_Codes, ExplicitMSR_Tamo, ExplicitMSR, ExplicitMSR_nearOptimal}.

The adopted models in the regenerating codes' literature are typically considering the \emph{rigid symmetric repair} in which a predetermined number $d\leq n-1$ of helpers participate in the repair scenario, each providing $\beta = \gamma/d$ repair symbols. However, for a distributed storage system, the capability to adapt to various conditions in which data reconstruction and repair could be performed increases the robustness significantly. The varying capability of helpers for participating in an specific repair is not only limited to load unbalance or varying quality of connection, and could rise due to many other practical reasons such as geographical distance among the nodes in the storage network or topological and infrastructure asymmetry of the network. Besides loosing access to storage nodes, it has been shown that data corruption is a common problem in large scale DSS's \cite{DataCorruption}.

In this work we consider two simultaneous extensions to the original formulation of Dimakis \emph{et. al.} \cite{Regenerating}. One extra feature we consider in our setup is that the number of helpers chosen for repair can be adaptively selected, which allows for run-time optimization according to the dynamic state of the system. We refer to this property as \emph{bandwidth adaptivity}. Such flexibility adds a notable robustness to distributed storage systems with time-varying conditions such as peer-to-peer distributed storage systems where storage nodes join and leave the system frequently. It is also an important feature for systems with significant load unbalance or heterogeneous connectivity where the number of available nodes for repair vary for different repair cases. The importance of this feature has been addressed by other researchers \cite{Flexible_Regenerating, Exact_AI_Asymptotic, BWAdapt_Kermarrec, BWAdapt_Opportunistic, ProgressiveEngagement, ExplicitMSR}. Yet none of these works have addressed the MBR exact repair case.

As the second extension to the original regenerating codes, in our model we also encounter the presence of error in the system by adopting an adversarial intruder which can compromise up to a fixed number $b < k/2$ of nodes, and following \cite{ErrRes_Salim_Journal,ErrRes_PM} we will refer to it as \emph{limited power adversarial intruder}. Such intruder is considered to be omniscient, i.e. knows the original data stored in the system and the coding scheme, and can control the data stored in, and being transferred by no more than $b$ nodes under his control.

In \cite{ErrRes_Salim_Journal} an upper bound for the capacity of distributed storage systems is presented in the presence of different intruders including the limited power omniscient intruder described above. Exact repair coding schemes are also proposed for MBR and MSR modes in \cite{ErrRes_Salim_Journal,ErrRes_PM} to achieve this upper bound. However, none of these results consider bandwidth adaptivity.


Note that the code design for bandwidth adaptive functional repair regenerating codes simply reduces to the network coding for single source multicast problem. However, for the exact repair regenerating codes, this problem has only been considered for the MSR case, and a solution is provided in \cite{Exact_AI_Asymptotic} based on interference alignment, which achieves optimality in asymptotically large per node storage capacity $\alpha$, and repair bandwidth $\beta$. Recently, a few other bandwidth adaptive exact repair MSR codes have been introduced which provide optimality with finite values for $\alpha$, and $\beta$ \cite{ExplicitMSR, BAMSR, BAMSR_ITW17}.

In this work we focus on the natural extension of MBR mode with bandwidth adaptivity and error resiliency, and present exact repair coding schemes which we show are optimal. To the best of our knowledge, this work is the first non-asymptotic exact repair code construction for such a setting. The main contributions of this work are explained in Section \ref{Sec_MainResults}, after briefly reviewing the related works and formally defining the setup in the next two sections.

 
\section{Background and Related Works}

\subsection{Background on Error Resiliency}\label{Sec_ErrRes}

The introduction of errors into the data stored or transmitted over the network has been the subject of interest for a long time. In applications such as the distributed storage systems, where the integrity of the stored data is the highest priority, any practical scheme has to encounter an appropriate mechanism for dealing with the issue of introducing and propagating errors. While the source of error can also be non-adversarial, such as random errors introduced naturally through the I/O process in the storage nodes or during the data transmission, in order to provide performance guarantee levels it is common to consider adversarial models for the error source. In this work we consider a malicious intruder controlling a subset of storage nodes in the system as considered in \cite{ErrRes_Salim_Journal, ErrRes_Salim_Conf, ErrRes_Polytope, ErrRes_PM, ErrRes_Byzentine_MultiRep, ErrRes_Byzentine_Infocom, ErrRes_Byzentine_TCOM}. The intruder model in these works, which has also been referred to as the "Byzentine intruder", will be referred to as \emph{"limited power omniscient adversary"} in this work. In this model, the adversarial intruder could control anything about the data transmitted from, or stored in a fixed number, namely $b$, of the storage nodes, and has complete knowledge of the coding scheme and the encoded message. The formal definition of a reliable and error resilient distributed storage system under such a limited power omniscient adversary is given as follows.

\begin{defn}[Error Resilient Regenerating Code]\label{dfn_BasicERRCodes}
An error resilient regenerating code, $\mathcal{C}(n,$ $k,$ $d,$ $b,$ $\alpha,$ $\gamma)$, is a regenerating code with $n$ nodes, among which at most $b$ nodes provide adversarial compromised data. The code should be capable of performing two operations; i) Genuine repair for any failed node, in which the genuine content of the failed node is regenerated by accessing any $d$ remaining nodes, and downloading $\beta=\gamma/d$ repair symbols from each of them. ii) Genuine data reconstruction, in which the whole stored data is reconstructed by accessing the content of any $k$ nodes.
\end{defn}

In \cite{ErrRes_Salim_Journal}, an upper bound was derived for the total storage capacity of an error resilient regenerating code, which is a network version of the Singleton bound \cite{Singleton2, Singleton}. Note that however, the model considered in \cite{ErrRes_Salim_Journal} does not require the repair procedure to be error-free, and hence allows propagation of error during the repair procedure. The following theorem rephrases this upper bound.

\begin{thm}\label{Thm_ErrRes_Cap_UB}\cite{ErrRes_Salim_Journal}
The total storage capacity of any error resilient regenerating code, $\mathcal{C}(n,$ $k,$ $d,$ $b,$ $\alpha,$ $\gamma)$, is upper bounded as follows.
\begin{align}\label{eq_ErrRes_Cap_UB}
F \leq \sum_{i=2b}^{k-1}{\min \left\lbrace \alpha,(d-i)\frac{\gamma}{d}\right\rbrace }. 
\end{align}
\end{thm}
%

Moreover, \cite{ErrRes_Salim_Journal} also provides an explicit exact repair code construction for the MBR case in the special case of $d = n-1$ based on the exact repair regenerating code introduced in \cite{Exact_R_by_T_1}, which achieves the total storage capacity upper bound presented in Theorem \ref{Thm_ErrRes_Cap_UB}. However, \cite{ErrRes_Salim_Journal} leaves this as an open problem whether or not the upper bound of (\ref{eq_ErrRes_Cap_UB}) is tight for exact repair error resilient regenerating codes in any other case. 

Rashmi \emph{et. al.} \cite{ErrRes_PM, ErrRes_PM_Journal} later included the error-free repair mechanism in the model considered in \cite{ErrRes_Salim_Journal}, and showed that the upper bound of (\ref{eq_ErrRes_Cap_UB}) is tight for the MBR exact repair regenerating codes with any choice of fixed parameters $k$, and $d\geq k$. They also proved the tightness of upper bound for MSR exact repair regenerating codes with $d \geq 2(k-1)$ based on the PM exact repair regenerating codes introduced in \cite{PM_Codes}. In the presented coding scheme for MBR and MSR cases in  \cite{ErrRes_PM_Journal}, the number of accessed nodes for repair and data reconstruction procedures can be chosen in the run-time. As a result by increasing the number of accessed nodes, and proportionally increasing the total required data transmission, the maximum number of erroneous nodes against which the procedure remains secure also increases. This concept is referred to as \emph{"on-demand security"}. However, this is different from the bandwidth adaptivity property, presented in our work, in which, the maximum number of erroneous nodes is predefined, and increasing the number of accessed nodes helps by reducing the total repair bandwidth.

In \cite{ErrRes_MRD_Silberstein}, the omniscient adversary is considered to be able to replace the content of an affected node only once, and the total storage capacity bounds and achievable schemes for this setting are provided. The paper also considers the MSR setting without restrictions on the number of times that adversary can compromise the contents of affected nodes, and provides schemes that are optimal for a specific choice of the parameters.

In \cite{ErrRes_Byzentine_Infocom, ErrRes_Byzentine_TCOM}, a similar setup is considered and a "progressive decoding" mechanism is introduced to reduce the computational complexity of repair and reconstruction procedures. These papers use a cyclic redundancy check (CRC) to confront errors in the repair and reconstruction procedures. However, CRC based schemes have this problem that the CRC may also be compromised by the adversary in the omniscient adversarial model. 

Many other researchers have also considered the error resiliency problem in regenerating codes along with other properties \cite{ErrRes_Byzentine_MultiRep, ErrRes_UpdateEfficient, ErrRes_LimittedKnowledge, ErrRes_Polytope, ErrRes_Hermit, ErrRes_RateMatch, ExplicitMSR}. For instance, \cite{ErrRes_Byzentine_MultiRep} considered multiple failure repair with cooperation among replacement nodes, and shows such cooperation could be adversely affective in the presence of an intruder. Efficiency in updating the stored data is considered in \cite{ErrRes_UpdateEfficient}, while \cite{ErrRes_LimittedKnowledge} considered limitations on the knowledge of the intruder. In \cite{ErrRes_Polytope} Polytope codes are used to provide error resiliency. Recently, \cite{ExplicitMSR, ErrRes_RateMatch} presented new error resilient coding schemes to address the maximum storage capacity in presence of an intruder. The schemes presented in \cite{ExplicitMSR} also achieve bandwidth adaptivity in the MSR case. In this work we consider bandwidth adaptivity along with the error resiliency, and focus on the coding schemes that achieve the minimum repair bandwidth.

%

\subsection{Background on Bandwidth Adaptive Repair}

The adopted models in the regenerating codes' literature are typically considering the \emph{rigid symmetric repair} in which a fixed number $d\leq n-1$ of helpers participate in the repair scenario, each providing $\beta = \gamma/d$ repair symbols. However, in practical systems, run-time optimization is of great interest. In particular, for a distributed storage system, the capability of adaptation to various conditions such that optimal repair procedure could be performed provides a lot of robustness. Such capabilities enables the system to maintain its functionality when some of the nodes are not able to provide enough repair bandwidth, e.g. due to being overloaded by an unbalanced load or due to temporarily loosing the quality of their links. In such conditions, it is valuable if the coding scheme is capable of performing the repair by accessing more helpers and receive less repair data from each. With such capability, we do not need to ignore the nodes with lowered capability of service and can still let them participate as much as they can, which could provide a potentially significant collective gain compared to the rigid setting. We will also show that accessing large number of helpers can even reduce not only the per node but also the total repair bandwidth, which in turn reduces the overall traffic load in such conditions. 

Similarly, we are interested to be able to perform repair based on the repair data provided by a small number of helpers as long as they are able to provide enough information to compensate for the absence of more helpers. Such conditions happen when a large number of nodes are inaccessible or have poor channel qualities but instead a few strong helpers are available. Then the coding scheme could ignore the poor helpers and expedite the repair based on the small group of strong helpers.

We will refer to a regenerating code capable of such flexibility in repair as a bandwidth adaptive regenerating code. Note that the dynamic capability of service for storage nodes is a well-known characteristic for many practical distributed systems such as peer-to-peer systems or heterogeneous network topologies \cite{SpaceMonkey, Tahoe, CFS, Google, OceanStore, TotalRecall}.

It was first in \cite{Flexible_Regenerating} that Shah \emph{et. al.} extend the regenerating code design problem in \cite{Regenerating_Conf, Regenerating} to include more flexibility. In \cite{Flexible_Regenerating} the authors consider the number of participating helpers to be selected independently in each repair or reconstruction. Moreover, they also relax the constraint of downloading the same amount of information from each node in both repair and reconstruction to allow \emph{asymmetric} participation of helpers, as long as all helpers contribute less than a fixed upper bound, $\beta_{\max}$, for repair. While the setting considered by Shah \emph{et. al.} provides much more flexibility in repair and reconstruction, the total repair bandwidth in their setting is always larger than that of the original regenerating codes formulations, except for the MSR case, where both settings achieve the same total repair bandwidth. Moreover, the coding scheme presented in \cite{Flexible_Regenerating} performs functional repair. 

The first work to address bandwidth adaptivity in the original setting of regenerating codes was \cite{BWAdapt_MFR}. In \cite{BWAdapt_MFR} Wang \emph{et. al.} introduced a functional repair coding scheme which works in the MSR mode and supports bandwidth adaptivity. Later \cite{BWAdapt_Kermarrec} also considered a similar setup and introduced functional repair MSR coding schemes with bandwidth adaptivity, while their main focus was on the derivation of the storage-repair-bandwidth trade-off for the functional \emph{coordinated} repair in regenerating codes. Note, however, that none of these works address the exact repair with bandwidth adaptivity in regenerating codes. 

Aggrawal \emph{et. al.} \cite{BWAdapt_Opportunistic} analysed the mean-time-to-failure (MTTF) in the regenerating codes with and without bandwidth adaptivity. This analysis is based on a birth-death process model in which the population of available storage node randomly changes with appropriately chosen rates. They showed that bandwidth adaptivity provides a significant gain in terms of MTTF. 
%

In exact repair regenerating codes, Cadambe \emph{et. al.} were the first to address the bandwidth adaptivity as an important property for regenerating codes in \cite{Exact_AI_Asymptotic}. They presented an exact repair coding scheme for the MSR mode with bandwidth adaptivity in the repair procedure, based on Interference Alignment. The code presented by Cadambe \emph{et. al.} is the first exact repair regenerating code with bandwidth adaptivity, however, their coding scheme only asymptotically achieves the optimal trade-off, when $\alpha$ and $\beta$ tend to infinity with proper ratio. Recently, bandwidth adaptive exact repair regenerating codes have been introduced for various parameters in the MSR case \cite{ExplicitMSR, BAMSR, BAMSR_ITW17}. In this work we focus on the MBR setting with error resiliency.
\section{Model and Results}\label{Section_Symmetric_Model}

\subsection{Model}
In this section we will briefly introduce the setup for a \emph{bandwidth adaptive and error resilient} (BAER) distributed storage system and the coding scheme of our interest. This model is a modified version of the original regenerating code's setup introduced in \cite{Regenerating_Conf, Regenerating}.

Throughout this work we consider a predefined finite filed, $\mathbb{F}_{q}$ of size $q$, as the code alphabet, such that all the symbols and operations through the network belong to $\mathbb{F}_q$.

\begin{defn}[BAER Regenerating Code and Flexibility Degree]
Consider the set of parameters $n$, $k$, $D=\{d_{1}, \cdots , d_{\delta} \}$, $b$, $\alpha$, and a \emph{total repair bandwidth function} $\gamma: D \rightarrow [\alpha,\infty)$. A BAER regenerating code $\mathcal{C}(n$, $k$, $D$, $b$, $\alpha$, $\gamma(\cdot))$ is a regenerating code with per node storage capacity $\alpha$, which performs repair and data reconstruction processes in a distributed storage network of $n$ nodes, when up to $b$ out of $n$ nodes are allowed to be erroneous, and provide adversarial data. Moreover, in any repair process the number of helpers, $d$, can be chosen arbitrarily from the set $D$. The choice of helper nodes is also arbitrary and each of the chosen helpers then provides $\gamma(d) / d$ repair symbols. Similarly, in any data reconstruction process the data collector accesses any arbitrary set of $k$ nodes and downloads $\alpha$ symbols from each. Moreover, the number of elements in the set $D$ is referred to as flexibility
degree of the code, and is denoted by $\delta$.
\end{defn}

\begin{rmrk}
Note that when erroneous nodes are present in the system the repair process should prevent the propagation of the errors. In other words, the repair of any node should replace that with a node that stores genuine data, which is the data that the coding scheme would have stored in the replacement node if no compromised node exists in the network. 
\end{rmrk}

\begin{defn}[Total Storage Capacity] 
For the set of parameters $n$, $k$, $D=\{d_{1}, \cdots , d_{\delta}\}$, $b$, $\alpha$, and a given function $\gamma: D \rightarrow [\alpha,\infty)$, the \emph{total storage capacity} of a BAER distributed storage system is the maximum size of the file that could be stored in a network of $n$ storage nodes with per node storage capacity $\alpha$, using a BAER regenerating code $\mathcal{C}(n$, $k$, $D$, $b$, $\alpha$, $\gamma(\cdot))$. We will denote the total storage capacity of such a system by $F(n$, $k$, $D$, $b$, $\alpha$, $\gamma(\cdot))$, or simply $F$, whenever the parameters of the system could be inferred from the context. 
\end{defn}

\begin{defn}[Optimal Codes]
Consider a set of parameters $n$, $k$, $D=\{d_{1}, \cdots , d_{\delta}\}$, $b$. For a given total storage capacity $F$, a BAER regenerating code $\mathcal{C}(n$, $k$, $D$, $b$, $\alpha$, $\gamma(\cdot))$ is optimal, if it realizes the total storage capacity $F$, and for any other BAER regenerating code $\mathcal{C}'(n$, $k$, $D$, $b$, $\alpha'$, $\gamma'(\cdot))$, realizing total storage capacity $F$, with,
\begin{align}
\alpha' \leq \alpha , \nonumber
\end{align}
we have
\begin{align}
\exists{d \in D},~ \gamma(d) < \gamma'(d). \nonumber
\end{align}
\end{defn}


\begin{rmrk}\label{Rmrk_no_per_node_redundancy}
It is obvious that for any optimal coding scheme, there exists no redundancy among the symbols stored in a single storage node. In other words, none of the symbols stored in a single storage node in an optimal BAER regenerating code can be calculated as a function of other symbols, otherwise, removing that symbol will reduce $\alpha$, while all the repair and data reconstruction procedures are still possible with the same data transmission as before. Through the rest of this work, we will only consider optimal BAER regenerating codes and hence assume there exists no per node redundancy.
\end{rmrk}

It worth to mention that optimal BAER codes do not necessarily exist for all sets of parameters $n$, $k$, $D=\{d_{1}, \cdots , d_{\delta}\}$, $b$, $\alpha$, $\gamma(\cdot)$. For example having a large $\alpha$ and a very small $\gamma(d)$ for some $d \in D$, makes genuine repair with $d$ helper impossible if no per node redundancy is allowed. This will be studied in details in the following sections and we will characterize the set properties of the sets of parameters for which optimal BAER codes exist.

Since there are many parameters involved in the presented setting, hereafter we will consider $n$, $k$, $D$, $b$, $\alpha$ to be fixed and mainly focus on exploring the tension between $F$ and $\gamma(\cdot)$. However, the results of this work still capture the overall trade-off between all the parameters. The following definitions set the scene for studying the tension between $F$ and $\gamma(\cdot)$, for fixed values of other parameters.

\begin{defn}[Set $\mathcal{O}_{\alpha}$, and Set $\Gamma_{\alpha}$]
For a given set of parameters $n$, $k$, $D=\{d_{1}, \cdots , d_{\delta}\}$, $b$, and a fixed per node storage capacity $\alpha$, we define the set of all optimal BAER regenerating codes with the same per node storage capacity $\alpha$ as $\mathcal{O}_{\alpha}$. Moreover, we denote the set of all the total repair bandwidth functions pertaining to optimal codes in $\mathcal{O}_{\alpha}$ by $\Gamma_{\alpha}$. In other words,
\begin{align}
\Gamma_{\alpha} = \{\gamma(\cdot) | \mathcal{C}'(n, k, D, b, \alpha, \gamma(\cdot))\in \mathcal{O}_{\alpha} \}. \nonumber
\end{align}
\end{defn}

Note that from Remark \ref{Rmrk_no_per_node_redundancy} we know that all codes in $\mathcal{O}_{\alpha}$ have zero per node redundancy. Moreover, the optimal codes in $\mathcal{O}_{\alpha}$ may have different total storage capacities, $F$, and total repair bandwidth functions, $\gamma(\cdot)$. The following definition addresses this set.

As mentioned above, for a given set of parameters $n$, $k$, $D$, $b$, and $\alpha$, we may have a set of optimal BAER regenerating codes  $\mathcal{O}_{\alpha}$, with different total repair bandwidth functions, $\gamma(\cdot)$. As a result, $\gamma$ can be optimized for any given $d\in D$. The natural question is whether the minimal values of $\gamma(d)$ can be simultaneously achieved for all values of $d\in D$ in an optimal BAER code. While optimal BAER regenerating codes correspond to the Pareto optimal functions $\gamma(\cdot)$, in this work we consider the strongest definition for optimality as will be introduced in the following definitions. Surprisingly, we will show that such strong optimality is achievable for the minimum repair bandwidth BAER regenerating codes.

\begin{defn}[MBR BAER Codes]\label{Def_MBR_mode}
For a set of parameters $n$, $k$, $D$, $b$, $\alpha$, the MBR BAER code is an optimal BAER regenerating code with total repair bandwidth function $\gamma_{\text{MBR}}(\cdot)$, such that
\begin{align}
\gamma_{\text{MBR}}(d) = \min_{\gamma(\cdot) \in \Gamma_{\alpha}}{\gamma(d)}. \nonumber
\end{align}
In other words, $\gamma_{\text{MBR}}(d)$ is the minimum possible repair bandwidth for all values of $d \in D$, among all optimal BAER codes with fixed parameters $n$, $k$, $D$, $b$, and $\alpha$. We also denote the total storage capacity associated with parameters $n$, $k$, $D$, $b$, $\alpha$, and $\gamma_{\text{MBR}}(\cdot)$, by $F_{\text{MBR}}(n$, $k$, $D$, $b$, $\alpha$, $\gamma_{\text{MBR}}(\cdot))$, or $F_{\text{MBR}}$ for short. Moreover, the BAER regenerating codes with parameters $n$, $k$, $D$, $b$, $\alpha$, and $\gamma_{\text{MBR}}(\cdot)$, realizing $F_{\text{MBR}}$ are referred to as MBR BAER codes.
\end{defn}

Note that since the $\gamma_{\text{MBR}}(\cdot)$ is defined to be the point-wise minimum of the total repair bandwidth functions in $\Gamma_{\alpha}$, there is no guarantee that there exists a non-trivial MBR BAER code for a given set of parameters. In this work we will show that non-trivial MBR BAER codes exists for a wide range of parameters.


\subsection{Main Results}\label{Sec_MainResults}

In this work we focus on the MBR mode with exact repair. Considering the set of parameters $n$, $k$, $b$, $\alpha$, and a set $D=\{d_{1}, \cdots, d_{\delta}\}$, for some flexibility degree $\delta > 1$, such that
\begin{align}\label{eq_d_i_Cond}
2 b < k \leq d_{1} \leq \cdots \leq d_{\delta},
\end{align}
and 
\begin{align}\label{eq_PernodeCap}
\alpha = \textrm{lcm}(d_{1} - 2 b , \cdots , d_{\delta} -  2 b) a,
\end{align}
for some integer $a \geq 1$. Through this work we will use the notation $d_{\min} = d_{1}$ to emphasise the fact that this is the smallest element is the set $D$.

\begin{rmrk}
Regarding the constraint (\ref{eq_PernodeCap}) on the per node storage capacity, note that in practice the per node storage capacity is usually very large, e.g. a few Terabytes, hence this constraint is not adding any practical limitation.
\end{rmrk}

We first characterize the minimum total repair bandwidth function such that the storage capacity is positive, and hence determine the MBR mode for the described setting. The following theorem describes this result.

\begin{thm}\label{THM_MBR_OptimalRBW}
Consider a BAER regenerating code with parameters $n$, $k$, $b$, $\alpha$, and the set $D=\{d_{1}, \cdots , d_{\delta}\}$ satisfying conditions (\ref{eq_d_i_Cond}), and (\ref{eq_PernodeCap}). Then assuming zero per node redundancy, the minimum total repair bandwidth function is given by 
\begin{align}\label{eq_MBR_Gamma_BAER_code}
\gamma_{\text{MBR}}(d) = \frac{\alpha d}{d-2b},~~ \forall{d \in D}.
\end{align}
\end{thm}

The proof of this result is provided in Section \ref{Sec_Converse}.

We also determine the total storage capacity of a BAER regenerating code in the MBR mode, and provide an upper bound for the total storage capacity of the general mode. We summarize these results in the following theorem.

\begin{thm}\label{THM_MBR_Capacity}
Consider a BAER regenerating code with parameters $n$, $k$, $b$, $\alpha$, and the set $D=\{d_{1}, \cdots , d_{\delta}\}$ satisfying conditions (\ref{eq_d_i_Cond}), and (\ref{eq_PernodeCap}). Using the notation $d_{\min}= d_{1}$, for any MBR BAER code, as introduced in Theorem \ref{THM_MBR_OptimalRBW} with the total repair bandwidth function, 
\begin{align}
\gamma_{\text{MBR}}(d) = \frac{\alpha d}{d-2b},~\forall{d\in D}, \nonumber
\end{align}
the total storage capacity is given by
\begin{align}\label{eq_MBR_Capacity}
F_{\text{MBR}} &= \sum_{j=0}^{k - 2 b -1}{(d_{\min} - 2 b - j)\frac{\alpha}{(d_{\min} - 2 b)}} \nonumber \\
&=\frac{\alpha}{d_{\min}-2 b}(k-2 b)\left(d_{\min}-b-\frac{k-1}{2}\right).
\end{align}
Moreover, the following upper bound holds for the total storage capacity of any BAER regenerating code, associated with the arbitrary total repair bandwidth function $\gamma(\cdot)$.
\begin{align}\label{eq_F_UpperBound}
F \leq \sum_{j=0}^{k - 2 b -1}{\min\left(\alpha,\min_{d\in D}\left((d - 2 b - j)\frac{\gamma(d)}{d}\right)\right)}.
\end{align}
\end{thm}

Note that the results in the above theorems hold for both the exact repair as well as the functional repair.

The upper bound (\ref{eq_F_UpperBound}) is derived based on two modifications in the standard information flow graph which was originally introduced in \cite{Regenerating}; \emph{i)} Considering a \emph{genie} assisting the decoder in each repair or data reconstruction by determining a subset of accessed nodes of size $2b$ which contains all the compromised nodes among the selected ones, and \emph{ii)} Allowing the number of helper nodes $d \in D$ to change independently in every term of the right-hand-side expression to minimize the resulting upper bound. The detailed proof is provided in Section \ref{Sec_Converse} and Appendix \ref{App_ConverseLemma}. 

In this wrok we provide two achievablity schemes in the MBR mode based on an extension of the MBR PM regenerating codes. These schemes are described in Section \ref{Sec_CodeConstruction}, and Section \ref{Sec_Coding_Sch_II}. More specifically, we show that universally optimal exact repair BAER regenerating code exists for the MBR mode by providing explicit code constructions. To the best of our knowledge the only other bandwidth adaptive exact repair regenerating code constructions presented so far are MSR codes \cite{Exact_AI_Asymptotic, ExplicitMSR, BAMSR_ITW17, BAMSR}. In this work we present two different schemes for the bandwidth adaptive and error resilient repair procedure in the MBR mode. Both schemes are capable of performing exact repair. The first scheme has a larger code alphabet field size requirement while the second scheme reduces the field size requirement to that of the MBR PM codes, while has larger per node storage capacity requirement.

\section{First Coding Scheme}\label{Sec_CodeConstruction}

For an arbitrary flexibility degree $\delta > 1$, we introduce an exact repair MBR BAER code $\mathcal{C}(n$, $k$, $D$, $b$, $\alpha$, $\gamma(\cdot))$, for $D = \{d_{1}, \cdots , d_{\delta}\}$, where the parameters satisfy the conditions given in (\ref{eq_d_i_Cond}), and (\ref{eq_PernodeCap}). This coding scheme achieves the total storage capacity $F_{\text{MBR}}$ as given in (\ref{eq_MBR_Capacity}). We also use the notation $d_{\min}$ to denote
\begin{align}
d_{\min} = \min_{d\in D}\{d\}=d_{1}. \nonumber
\end{align}


The presented code construction could be considered as a generalization of the \emph{Product Matrix} (PM) MBR codes introduced by Rashmi \emph{et. al.} \cite{PM_Codes}, in which we use the PM codes as basic components. However, the repair scheme is properly redesigned to provide bandwidth adaptivity and error resiliency properties as defined in Section \ref{Section_Symmetric_Model}. We use examples throughout the discussions to better illustrate the ideas and proposed encoding/decoding algorithms. For the rest of this section let $\underline{s}=[s_1, \cdots, s_{F_{\text{MBR}}}]$ denote the source data symbols, where $F_{\text{MBR}}$ is given in (\ref{eq_MBR_Capacity}).

As mentioned before, we will consider all the symbols and operations to belong to a Galois filed $\mathbb{F}_{q}$, referred to as the code alphabet. In particular we consider $\mathbb{F}_{q}$ to be simple extension filed over a base field $\mathbb{F}_{p}$ for a large enough prime $p$. Let $g \in \mathbb{F}_{q}$ denote the primitive element of the code alphabet, and $\mathbb{F}_{p}[x]$ denote the ring of polynomials with coefficients from $\mathbb{F}_{p}$, and $\varrho(x)\in \mathbb{F}_{p}[x]$, denote the minimal polynomial of $g$, then we have,
\begin{align}
\mathbb{F}_{q} \simeq \mathbb{F}_{p}[x] / \langle \varrho(x) \rangle. \nonumber
\end{align} 

We provide a discussion on the field size requirement for the coding scheme at the end of this section, after introducing the coding scheme.

\subsection{Construction for $b=0$}\label{Subsec_concatenation}

In this subsection we assume there is no compromised node in the network, \emph{i. e.,} no node provides erroneous data. We show that in this case an MBR bandwidth adaptive code $\mathcal{C}(n$, $k$, $D$, $b=0$, $\alpha$, $\gamma(\cdot))$, for an arbitrary flexibility degree $\delta$, is achievable through a simple concatenation of $\alpha / d_{\min}$ MBR PM codes, with parameters $n'=n$, $k'=k$, $d'=d_{\min}$, $\alpha'=d_{\min}$, $\beta' = 1$, as in \cite{PM_Codes}. We will refer to these MBR PM codes as \emph{component codes}. More specifically, we choose each component code to be an MBR PM code capable of performing data reconstruction using $k$ nodes, and exact repair using $d_{\min}$ helpers only. This simple scheme then provides exact repair and bandwidth adaptivity, but it is not error resilient. In the following subsections we will show how we can achieve an MBR BAER exact-repair coding scheme by modifying this naive concatenation scheme. 

For the encoding of the concatenated scheme, we first partition the source data symbols $\underline{s}=[s_1, \cdots, s_{F_{\text{MBR}}}]$ into $\alpha / d_{\min}$ disjoints partitions, each of size 
\begin{align}
\frac{F_{\text{MBR}}}{\left(\frac{\alpha}{d_{\min}}\right)}=\frac{k(k+1)}{2}+k(d_{\min}-k), \nonumber
\end{align}
where $F_{\text{MBR}}$ is given in (\ref{eq_MBR_Capacity}). We encode each partition separately for the set of $n$ storage nodes using one of the component codes to produce $d_{\min}$ encoded symbols to be stored on each storage node. Then concatenating the encoded symbols provided by all the $\alpha / d_{\min}$ component codes for each node, we form a vector of size $\alpha$ encoded symbols to be stored in each node. The following figure illustrates this process.

\begin{figure}
\centering
\resizebox{3.5 in}{!}{\includegraphics{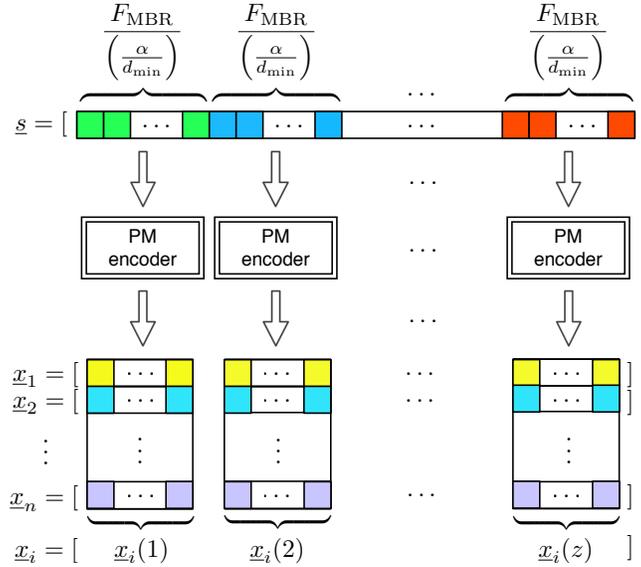}}
\caption{The encoding process for the simple concatenation scheme.}\label{Fig_Concatenation}
\end{figure}

The following example illustrates this procedure. 

\begin{example}\label{EX_Naive}
Let the number of nodes be $n = 5$, $\delta = 2$ and set $k = 2$, $d_{\min}=d_{1}=3$, $d_{2}=4$, and $b = 0$. To satisfy (\ref{eq_PernodeCap}) let $\alpha = 12$, and from (\ref{eq_MBR_Capacity}) we have $F_{\text{MBR}}=20$. We will use the Galois filed $\mathbb{F}_{7}$ as the code alphabet. Let the source symbols be $\underline{s}=[s_1, \cdots, s_{20}]$. We partition $\underline{s}$ into $\alpha/d_{\min}=4$ partitions each containing $5$ source symbols as $\underline{s}(1)=[s_{1},\cdots,s_{5}], \cdots, \underline{s}(4)=[s_{16},\cdots,s_{20}]$, and encode each partition separately. To this end, we then form data matrices $M_{1},\cdots,M_{4}$ based on the symbols in $\underline{s}(1),\cdots,\underline{s}(4)$ respectively,
\begin{align}
M_{1}=\left[\begin{array}{c c c}
s_{1}~&~s_{2}~&~s_{4} \\
s_{2}~&~s_{3}~&~s_{5} \\
s_{4}~&~s_{5}~&~0~
\end{array}\right], \cdots , M_{4}=\left[\begin{array}{c c c}
s_{16}~&~s_{17}~&~s_{19} \\
s_{17}~&~s_{18}~&~s_{20} \\
s_{19}~&~s_{20}~&~0~
\end{array}\right]. \nonumber
\end{align}
Next, we pick a Vandermonde matrix of size $n \times d_{\min} = 5 \times 3$,
\begin{align}
\Psi = \left[\begin{array}{c c c}
1~&~e_{1}~&~(e_{1})^2 \\
1~&~e_{2}~&~(e_{2})^2 \\
~~& ~\vdots ~&~~ \\
1~&~e_{5}~&~(e_{5})^2
\end{array}\right] = \left[\begin{array}{c}
\underline{\psi}_{1} \\
\underline{\psi}_{2} \\
\vdots \\
\underline{\psi}_{5}
\end{array}\right], \nonumber
\end{align}
where $e_{j}$'s are distinct non-zero elements of the code alphabet $\mathbb{F}_{7}$. Then following the encoding mechanism of PM MBR codes, the vector of encoded symbols to be stored on some storage node $\ell \in \{1,\cdots , 5\}$ is obtained as
\begin{align}
\underline{x}_{\ell} = [\underline{x}_{\ell}(1), \underline{x}_{\ell}(2), \underline{x}_{\ell}(3), \underline{x}_{\ell}(4)],  \nonumber
\end{align}
where,
\begin{align}
\underline{x}_{\ell}(i) = \underline{\psi}_{\ell} M_{i}. \nonumber
\end{align}
\end{example}

When a storage node fails, in order to perform a repair, we perform $\alpha/d_{\min}$ separate repairs, each corresponding to one component code. Repair of each component code requires $d_{\min}$ helpers and one repair symbol from each helper. The set of helpers for any two different component codes could have an intersection of arbitrary size between zero and $d_{\min}$. The flexibility required in the choice of the number of helpers comes from the choice of the size of intersections between the sets of helpers for different code components. 

The following lemma guarantees this concatenated coding scheme can perform a symmetric repair with $d$ helpers and per node repair bandwidth $\gamma(d)/d$, and hence is bandwidth adaptive for $d\in D$.

\begin{lem}
For any $d\in D$, the concatenated coding scheme is capable of performing symmetric repair with $d$ helpers, and per node repair bandwidth
\begin{align}
\beta(d)=\frac{\gamma(d)}{d}=\frac{\alpha}{d}. \nonumber
\end{align}
Recall that (\ref{eq_PernodeCap}) guarantees that $\beta(d)$ is integer for all $d\in D$.
\end{lem}

\begin{proof}
Any node in the network has one coded segment for each MBR PM component code, and therefore could serve as a helper, providing one repair symbol to repair the lost coded segment in that component code. The repair will be performed separately for each of the $\alpha/d_{\min}$ MBR PM component codes, and each component code's repair requires $d_{\min}$ repair symbols from distinct helpers. Then we only need to show there exists an assignment of code components to the helpers such that each helper is assigned to exactly $\alpha/d$ code components, and each code component is assigned to exactly $d_{\min}$ distinct helpers. 

To guarantee the possibility of this assignment we simply introduce an assignment bipartite graph with $\alpha/d_{\min}$ vertices on the left denoted by $\mathcal{V}$ (representing the code components) and $d$ vertices on the right denoted by $\mathcal{U}$ (representing the selected helpers). It is enough to show there exists a regular bipartite graph with left degree $d_{\min}$ and right degree $\beta(d)=\alpha / d$ for each choice of $d\in D$. Algorithm \ref{Alg_AssignmentBipartite} creates such a bipartite graph using the ideas introduced in the well-known Havel-Hakimi algorithm \cite{Havel, Hakimi}. The proof of correctness of this algorithm is provided in Appendix \ref{App_Alg_Bipartite}.
\begin{algorithm}
\caption{Bipartite graph construction}\label{Alg_AssignmentBipartite}
\begin{algorithmic}[1]
\State Input: Two sets of vertices $\mathcal{V},~\mathcal{U}$, and parameter $d_{\min}$.
\State Set $d=|\mathcal{U}|$.
\State Initiate a bipartite graph with vertex sets $\mathcal{V},~\mathcal{U}$ and no edges.
\For{each vertex $v\in \mathcal{V}$}
\State Create $\mathcal{W}\subset \mathcal{U}$ of $d_{\min}$ vertices in $\mathcal{U}$ with least degrees.
\State Connect $v$ to all the vertices in $\mathcal{W}$
\EndFor
\end{algorithmic}
\end{algorithm}

\end{proof}

Note that the total repair bandwidth is always $\alpha$, which is the minimum possible total bandwidth required to repair a node of capacity $\alpha$ and hence the bandwidth adaptivity in this naive concatenation based scheme is achieved at no extra cost for every choice of $d\in D$. Since MBR PR codes are optimal then this bandwidth adaptive version is also MBR and universally optimal for all $d\in D$, though it is not yet error resilient.

\begin{example}\label{EX_Naive2}
Let us continue with the code in Example \ref{EX_Naive}, with $n = 5$, $\delta = 2$, $k = 2$, $d_{\min}=d_{1}=3$, $d_{2}=4$, and $b = 0$. Assume that node $5$ is failed and being replaced through a repair and there are four other nodes still available in the network, namely nodes $1,2,3,4$. In this example there are two options for $d$, the number of helpers. One option is to chose $d=d_{\min}=3$ and select an arbitrary subset of size three from the available storage nodes. Without loss of generality, assume the set of selected helpers is $\mathcal{H}_{1}=\{1,2,3\}$. Each helper node $h \in \mathcal{H}_{1}$ sends $\alpha/d_{\min}=4$ repair symbols and each component code needs $d_{\min}=3$ distinct helpers to provide one repair symbol each. The helper assignment graph in this case is the complete bipartite graph with $|\mathcal{V}|=4$, and $|\mathcal{U}|=3$, in which every node in $\mathcal{V}$ is connected to every node in $\mathcal{U}$, as depicted in Fig. \ref{Fig_BipartiteGraph} (a). Therefore, each helper will provide one repair symbol for every component code.

Another choice in this example is to select $d=d_{2}=4$. The selected set of helpers is then $\mathcal{H}_{2}=\{1,2,3,4\}$. Since we have $\alpha/d_{\min}=4$ code components, we denote the set of vertices in $\mathcal{V}$ by $\{v_1, \cdots , v_{4}\}$, where $v_i$ represents the $i^{\text{th}}$ code component. The vertices in $\mathcal{U}$ are denoted by $\{u_1,\cdots,u_4\}$, where $u_{h}$ represents the helper node $h$. Moreover, let's represent the set of neighbours of a vertex $v_i$ by $\mathcal{N}(v_i)$. Then one possible realization of the helper assignment bipartite graph created by Algorithm \ref{Alg_AssignmentBipartite} is represented by $\mathcal{N}(v_1)=\{u_{1},u_{2},u_{3}\}$, $\mathcal{N}(v_2)=\{u_{1},u_{2},u_{4}\}$, $\mathcal{N}(v_3)=\{u_{1},u_{3},u_{4}\}$, and finally $\mathcal{N}(v_4)=\{u_{2},u_{3},u_{4}\}$, as depicted in Fig. \ref{Fig_BipartiteGraph} (b). Then for instance the helper node $1$ will send three repair symbols including one for component code $1$, one for component code $2$, and one for component code $3$, but no repair symbol for component code $4$. Similarly, the encoded symbols on the failed node $5$, pertaining to the component code $2$, namely $\underline{x}_{5}(2)$, will be repaired based on the corresponding repair symbols received from helper nodes $1$, $2$, and $4$.

\begin{figure}
\psfragscanon
\psfrag{A}{$v_{1}$}
\psfrag{B}{$v_{2}$}
\psfrag{C}{$v_{3}$}
\psfrag{D}{$v_{4}$}
\psfrag{E}{$u_{1}$}
\psfrag{F}{$u_{2}$}
\psfrag{G}{$u_{3}$}
\psfrag{H}{$u_{4}$}
\psfrag{I}{$\mathcal{V}$}
\psfrag{J}{$\mathcal{U}$}
\psfrag{K}{$(\text{a})$}
\psfrag{L}{$(\text{b})$}
\centering
\resizebox{3 in}{!}{\includegraphics{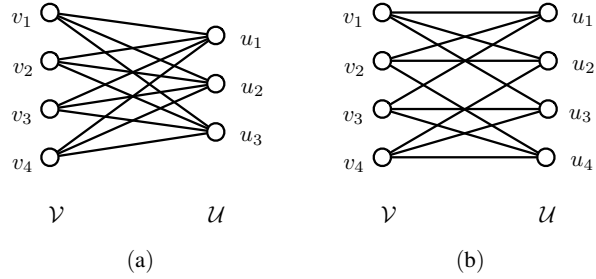}}
\caption{The bipartite assignment graphs generated by Algorithm \ref{Alg_AssignmentBipartite} for Example \ref{EX_Naive2}.}\label{Fig_BipartiteGraph}
\end{figure}

\end{example}

In order to perform data reconstruction, we can simply download the content of any $k$ nodes in the system and this will provide each component code with the required data to perform separate data reconstructions.


\subsection{Encoding for Storage with $b>0$}\label{Subsec_Enc_for_Storage}

In this subsections we describe an MBR bandwidth adaptive and error resilient (BAER) coding scheme with parameters $n$, $k$, $D=\{d_{1}, \cdots , d_{\delta}\}$, for an arbitrary flexibility degree $\delta$, and per node storage capacity $\alpha$ such that (\ref{eq_d_i_Cond}) and (\ref{eq_PernodeCap}) are satisfied. This scheme is capable of performing exact repair while there are up to $b$ erroneous nodes in the network. 

Let 
\begin{align}\label{eq_lambda_kappa_def}
\lambda=d_{\min}-2 b, ~~\kappa=k-2 b,
\end{align} 
and 
\begin{align}
z = \frac{\alpha}{\lambda}. \nonumber
\end{align}
Also let $O$ be a $\lambda \times \lambda$ zero matrix. The first step in the encoding process is to partition the source data symbols $\underline{s}=[s_1, \cdots, s_{F_{\text{MBR}}}]$ into $z$ disjoints partitions, and arrange the source data symbols of the $i^{\text{th}}$ partition in the form of the data submatrix $M_{i},~i \in \{1,\cdots , z \}$. Each of the data submatrices $M_{i},~i \in \{1,\cdots , z\}$ is a symmetric matrix satisfying the structural properties of data matrix in a Product Matrix MBR code, capable of performing exact repair using $\lambda$ helpers, and data reconstruction by accessing $\kappa$ nodes. In other words, 
\begin{align}
M_{i}=\left[ \begin{array}{l r}
N_{i} & L_{i} \\
L_{i}^\intercal & O'
\end{array}\right],~~i \in \{1,\cdots , z\}, \nonumber
\end{align}
where $N_{i}$ is a symmetric $\kappa \times \kappa$ matrix, and $L_{i}$ is a $\kappa \times (\lambda-\kappa)$ matrix, and finally $O'$ is a $(\lambda-\kappa) \times (\lambda-\kappa)$ zero matrix. Moreover, $\intercal$ denotes matrix transpose.

The overall \emph{data matrix} $M \in \mathbb{F}_{q}^{\alpha \times \alpha}$, is then formed as a block diagonal matrix consisting of $z$ submatrices $M_{1}, \cdots , M_{z}$ as the diagonal blocks in the following form.
\begin{align}\label{eq_Information_Matrix}
M = \left[ \begin{array}{c c c c}
M_{1} & O & \cdots & O\\
O & M_{2} & \cdots & O \\
\vdots & \vdots & \ddots & \vdots \\
O & \cdots & O & M_{z}
\end{array}\right].
\end{align}

Similar to the original  MBR Product Matrix codes the encoding for storage over each node is performed using a node-specific \emph{coefficient vector}. The coefficient vector in this construction is a row of an $n \times \alpha$ Vandermonde matrix $\Psi$,
\begin{align}\label{eq_Psi_def}
\Psi = \left[ \begin{array}{c c c c c}
1 & e_{1} & (e_{1})^{2} & \cdots & (e_{1})^{\alpha-1} \\
1 & e_{2} & (e_{2})^{2} & \cdots & (e_{2})^{\alpha-1} \\
\vdots & \vdots & \vdots & \ddots & \vdots\\
1 & e_{n} & (e_{n})^{2} & \cdots & (e_{n})^{\alpha-1}
\end{array}\right],
\end{align}
where, $e_{i},~i\in\{1,\cdots , n\}$ are distinct non-zero elements of $\mathbb{F}_{q}$. Without loss of generality we simply use
\begin{align}\label{eq_e_i_def}
e_{i} = g^{i},
\end{align}
where $g$ is the primitive element of $\mathbb{F}_{q}$.

Let's denote the $\ell^{\text{th}}$ row of $\Psi$ by $\underline{\psi}_{\ell}$, which is the coefficient vector for the $\ell^{\text{th}}$ storage node.The vector of encoded symbols to be stored on node $\ell \in \{1,\cdots,n\}$, denoted by $\underline{x}_{\ell}$ is obtained as
\begin{align}
\underline{x}_{\ell} = \underline{\psi}_{\ell}M. \nonumber
\end{align}

Note that $\underline{x}_{\ell}$ is an $\alpha$ dimensional vector and hence the per node storage capacity is satisfied. Also, the encoded vector for each node could be considered as concatenation of the encoded vectors of all the $z$ MBR PM code components with data matrices $M_{1}, \cdots , M_{z}$. To show this, we introduce the following partitioning for each node-specific coefficient vector $\underline{\psi}_{\ell}$ as 
\begin{align}\label{eq_psi_partition}
\underline{\psi}_{\ell} = [\underline{\psi}_{\ell}(1), \cdots , \underline{\psi}_{\ell}(z)],~~\ell \in \{1, \cdots , n\}, 
\end{align}
where, each segment $\underline{\psi}_{\ell}(i)$ is a $1 \times \lambda$ vector. Then for each node $\ell$ we have,
\begin{align}
\underline{x}_{\ell} = [\underline{\psi}_{\ell}(1)M_{1}, \cdots , \underline{\psi}_{\ell}(z)M_{z}]. \nonumber
\end{align}
 
\begin{example}\label{Ex_BAER_Symm_P1}
Consider the following set of parameters $\delta = 2$, $n = 6$, $k = 3$, $d_{\min} = d_{1} = 4$, $d_{2} = 5$, $b = 1$, and the code alphabet in use is $\mathbb{F}_{q}$, with primitive element $g$. Assume $\alpha=6$ to satisfy (\ref{eq_PernodeCap}). We have $\lambda=2,~\kappa=1$, and the number of component codes is $z=\alpha/\lambda=3$. Moreover, from (\ref{eq_MBR_Capacity}), we have $F_{\text{MBR}}=6$. Let's assume the source symbols are $\underline{s}=[s_{1}, \cdots , s_{6}]$. Finally, the data matrix $M$, and the corresponding Vandermonde matrix $\Psi$ will be
\begin{align}
M=\left[\begin{array}{c c c}
M_{1}&O&O \\
O&M_{2}&O \\
O&O&M_{3}
\end{array} \right]=\left[\begin{array}{c c | c c | c c}
s_{1}~&~s_{2}~&~0~&~0~&~0~&~0 \\
s_{2}~&~0~&~0~&~0~&~0~&~0 \\
\hline 
0~&~0~&~s_{3}~&~s_{4}~&~0~&~0 \\
0~&~0~&~s_{4}~&~0~&~0~&~0 \\
\hline 
0~&~0~&~0~&~0~&~s_{5}~&~s_{6} \\
0~&~0~&~0~&~0~&~s_{6}~&~0 
\end{array} \right], \nonumber
\end{align}
and 
\begin{align}
~~\Psi = \left[\begin{array}{c c c c c}
1~&~g~&~g^{2}~& \cdots &~g^{5} \\
1~&~g^{2}~&~\left(g^{2}\right)^{2}~& \cdots &~\left(g^{2}\right)^{5}\\
\vdots~&~\vdots~&~\vdots~&~\ddots~&~\vdots \\
1~&~g^{6}~&~\left(g^{6}\right)^{2}~& \cdots &~\left(g^{6}\right)^{5}
\end{array}\right]. \nonumber
\end{align}
Then, for instance, for node $1$ the coefficient vector is
\begin{align}
\underline{\psi}_{1} = [1, g, g^2, \cdots , g^{5}], \nonumber
\end{align}
and the encoded vectors is
\begin{align}
\underline{x}_{1} =\underline{\psi}_{1}M =[(s_{1}+g s_{2}), s_{2}, (g^{2} s_{3}+g^{3} s_{4}), g^{2} s_{4}, (g^{4} s_{5}+g^{5} s_{6}), g^{4} s_{6}]. \nonumber 
\end{align}
\end{example}

%
\subsection{Data Reconstruction with $b>0$}\label{Sec_Data_Recon}

The decoding procedure for both data reconstruction as well as the repair is performed using a scheme which will be referred to as the \emph{"test-group decoding"}. This decoding scheme enables the decoder to both recover the required data, and simultaneously authenticate the ingenuity of the recovered message. We describe the test-group decoding as an iterative procedure. In data reconstruction, the data collector accesses a set $\mathcal{S}$ of $k$ storage nodes. 

Each iteration uses one \emph{test-group} $\mathcal{T}$, which is a distinct subset of $\mathcal{S}$, consisting $k-b$ nodes, and the entire process continues for at most $\binom{k}{k-b}$ iterations. For a given iteration with a test-group $\mathcal{T}$, we examine all of its $\binom{k-b}{k-2b}$ subsets $\mathcal{H}$ of size $k-2b$, and perform data reconstruction using the data provided by nodes in $\mathcal{H}$. This gives an estimate for the data matrix $M$, which we denote by $\hat{M}_{\mathcal{H}}$. Recall that the code components are designed for parameters $\kappa=k-2b$ and $\lambda=d-2b$, and hence data reconstruction from nodes in $\mathcal{H}$ can be performed as for standard product matrix codes \cite{PM_Codes}.

In order to calculate an estimate for data reconstruction, based on $\mathcal{H}\subset \mathcal{T},~|\mathcal{H}|=k-2b$ note that each storage node contains $z$ coded segments, each pertaining to one component code. Therefore, the data collector could easily perform the data reconstruction for each of the $z$ component codes separately based on the corresponding segments provided by nodes in $\mathcal{H}$ using the same method introduced for the original Product Matrix codes \cite{PM_Codes}. As mentioned in Section \ref{Subsec_Enc_for_Storage}, each component code is designed to perform reconstruction based on $\kappa = k-2b$ storage nodes. Then assuming $\mathcal{H}$ is compromised-free, the corresponding segments of the codewords provided by the $k-2b$ nodes in $\mathcal{H}$ are sufficient for reconstruction in each component code.


Once all the estimates are calculated for a test-group, the decoder proceeds by checking the consistency of the estimates. The decoder stops whenever it finds a test-group such that all its estimates are  consistent and outputs the consistent estimate as the decoded data. The following lemma guarantees this procedure will always succeed.

\begin{lem}\label{Lem_Consistency}
Assuming the maximum number of compromised nodes is $b$, if all the estimates produced in the test-group decoding for a single test-group $\mathcal{T}$ are consistent, then all of them are correct. Moreover, the test-group decoding will always find a consistent test-group.
\end{lem}

\begin{proof}
First note that any test-group $\mathcal{T}$ consists of $k-b$ nodes and the decoder produces an estimate based on any subset of $\mathcal{H}$ of size $k-2b$ in $\mathcal{T}$. Since the maximum number of compromised nodes is $b$, then at least one of the subsets in any test-group $\mathcal{T}$ is guaranteed to be compromised-free. Therefore, at least one of the estimates in each test-group is genuine and if all the estimates in the test-group are consistent then all of them should match with the genuine estimate.

Now to prove that the test-group decoder always finds a consistent test-group, simply note that the maximum number of compromised nodes is $b$. Hence, there exists at least one subset of size $k-b$, in any set of $k$ selected nodes, $S$, which does not contain any compromised node. Since the test-group decoder uses all possible choices of test-groups before it terminates, it always processes a compromised-free test-group, which is consistent. 
\end{proof}

The test-group decoding for data reconstruction is summarized in Algorithm \ref{Alg_reconstriction}.

\begin{algorithm}
\caption{Test-group decoding for data reconstruction}\label{Alg_reconstriction}
\begin{algorithmic}[1]
\State Inpot: $k$, $b$, and $\underline{x}_{\ell}$ for all accessed node.
\State $\text{Consistency} \gets \text{False}$
\While{$\neg(\text{Consistency})$}
   \State $\mathcal{T}\gets \text{A new subset of helpers of size}~k-b$
   \For{each subset $\mathcal{H}\subset \mathcal{T}$}
   \State Calculate $\hat{M}_{\mathcal{H}}$
   \EndFor
   \If{$\hat{M}_\mathcal{H}$ matrices are the same $\forall{\mathcal{H}\subset \mathcal{T}}$}
   \State $\text{Consistency} \gets \text{True}$
   \State $\text{Output} \gets \hat{M}_\mathcal{H}~\text{for some}~\mathcal{H}\subset \text{consistent}~\mathcal{T}$   
   \EndIf
\EndWhile
\end{algorithmic}
\end{algorithm}

%

\subsection{Repair Scheme with $b>0$}\label{Subsec_Rep_Sch_I}

\subsubsection{Encoding for Repair with $b>0$}
In order to perform the encoding for a repair process, we use a Vandermonde matrix $\Omega \in \mathbb{F}_q^{z \times z}$, with $z = \alpha/(d_{\min}-2b)$, as,
\begin{align}\label{eq_Omega}
\Omega = \left[\begin{array}{c c c c}
\left(g^{i_{1}} \right)^{0}~ & ~\left(g^{i_{1}} \right)^{1}~ & ~\cdots~ & ~\left(g^{i_{1}} \right)^{z-1} \\
\left(g^{i_{2}} \right)^{0}~ & ~\left(g^{i_{2}} \right)^{1}~ & ~\cdots~ & ~\left(g^{i_{2}} \right)^{z-1} \\
\vdots & \vdots & \ddots & \vdots \\
\left(g^{i_{z}} \right)^{0}~ & ~\left(g^{i_{z}} \right)^{1}~ & ~\cdots~ & ~\left(g^{i_{z}} \right)^{z-1} \\
\end{array} \right],
\end{align}
where, 
\begin{align}
i_{1} < i_{2} < \cdots < i_{z}, \nonumber
\end{align}
and, for any $\ell_{1}, \ell_{2}$ such that $1 \leq \ell_{1}<\ell_{2}\leq z$ we have,
\begin{align}
i_{\ell_{2}}-i_{\ell_{1}} > \alpha n. \nonumber
\end{align}
This matrix $\Omega$ will be used to adjust the dimension of the vector of repair symbols provided by each helper based on the selected parameter $d$.

When a node $f$ in the network fails, we can choose any subset $\mathcal{H}$ of size $d$ of other nodes as helpers such that $2b < k \leq d_{\min} \leq d \leq d_{\delta}$. To describe the encoding process for repair at helper node $h \in \mathcal{H}$, we define the following notations. 
Let
\begin{align}
z_{d} = \frac{\alpha}{d-2b} \leq z. \nonumber
\end{align}
and, $\Omega_{z_{d}}$ denote the submatrix of $\Omega$, consisting of the first $z_{d}$ columns. Note that $\Omega_{z_{d}}$ is a Vandermonde matrix for any $d\in D$. Moreover, for any node $\ell \in \{1,\cdots,n\}$, let $\Phi_{\ell}$ denote an $\alpha \times z$ block-diagonal matrix with $\lambda \times 1$ diagonal blocks $\underline{\psi}_{\ell}^{\intercal}(1),\cdots , \underline{\psi}_{\ell}^{\intercal}(z)$ as
\begin{align}\label{eq_PhiMatrix}
\Phi_{\ell} = \left[ \begin{array}{c c c c}
\underline{\psi}_{\ell}^{\intercal}(1) ~ & ~ \underline{o}^{\intercal}_{\lambda} ~ & ~\cdots ~ & ~ \underline{o}^{\intercal}_{\lambda} \\
\underline{o}^{\intercal}_{\lambda} ~ & ~ \underline{\psi}_{\ell}^{\intercal}(2) ~ & ~\cdots ~ & ~ \underline{o}^{\intercal}_{\lambda} \\
\vdots ~ & ~ \vdots ~ & ~\ddots ~ & ~\vdots \\
\underline{o}^{\intercal}_{\lambda} ~ & ~ \cdots ~ & ~ \underline{o}^{\intercal}_{\lambda} ~ & ~ \underline{\psi}_{\ell}^{\intercal}(z)
\end{array}\right],
\end{align}
where $\underline{o}^{\intercal}_{\lambda}$ denotes a $\lambda \times 1$ zero vector.

Each helper node $h \in \mathcal{H}$ then produces a vector $\underline{r}_{h,f}$ of repair symbols for the failed node $f$ as 
\begin{align}\label{eq_RepairVector}
\underline{r}_{h,f}=\underline{x}_{h} \Phi_{f} \Omega_{z_{d}}.
\end{align}
Note that this will be a row vector of length $\alpha/(d-2b)$, then the per node repair bandwidth corresponding to the chosen parameter $d$ is
\begin{align}\label{eq_repairBW_Schm1}
\beta(d) = \frac{\gamma(d)}{d} = \frac{\alpha}{d - 2 b}.
\end{align}

\begin{example}\label{Ex_BAER_Symm_P2}
Let's continue the Example \ref{Ex_BAER_Symm_P1} with $\delta = 2$, $n = 6$, $k = 3$, $d_{\min}=d_{1} = 4$, $d_{2}=5$, $b=1$, $\alpha = 6$ and $z=3$. Assuming node $f=1$ is failed and choosing $d=d_{\min}=4$, we have 
\begin{align}
\Phi_{1}=\left[\begin{array}{c c c}
1~&~0~&~0 \\
g~&~0~&~0 \\
0~&~g^2~&~0 \\
0~&~g^3~&~0 \\
0~&~0~&~g^4 \\
0~&~0~&~g^5 
\end{array}\right], ~~\Omega_{z_{d_{\min}}}=\Omega = \left[\begin{array}{c c c}
1~&~g^{i_{1}}~&~\left(g^{i_{1}} \right)^{2} \\
1~&~g^{i_{2}}~&~\left(g^{i_{2}} \right)^{2} \\
1~&~g^{i_{3}}~&~\left(g^{i_{3}} \right)^{2} 
\end{array}\right]. \nonumber
\end{align}
Similarly, if we choose $d=d_{2}=5$, we have 
\begin{align}
\Omega_{z_{d_{2}}}=\left[\begin{array}{c c c c}
1~&~g^{i_{1}} \\
1~&~g^{i_{2}} \\
1~&~g^{i_{3}}
\end{array}\right], \nonumber
\end{align}

\end{example}



\subsubsection{Decoding for Repair with $b>0$}\label{Subsec_Dec_for_Rep_Rec}

On the decoder side, again the test-group decoder iteratively selects a test-group $\mathcal{T}$ of size $d-b$, which has not been used before. Then for any subset $\mathcal{H}\subset \mathcal{T}$ of size $d-2b$, the decoder calculates an estimate $\underline{\hat{x}}_{\mathcal{H},f}$ for the lost coded vector $\underline{x}_{f}$, and checks if all the calculated estimates match in the current test-group. If there is an inconsistency, the decoder terminates this iteration and starts the next iteration by selecting a new test-group until it finds a consistent one. An estimate in a consistent test-group is considered as the output.


Next we describe the process of calculating the estimate $\underline{\hat{x}}_{\mathcal{H},f}$, based on one arbitrary subset $\mathcal{H} = \{h_{1}, \cdots , h_{d-2b}\}$. First note that for each $h \in \mathcal{H}$, assuming it is not a compromised node, we have
\begin{align}\label{eq_middle_step}
\underline{r}_{h,f}&= \underline{x}_{h}\Phi_{f}\Omega_{z_{d}} \nonumber \\
&=\left[\underline{\psi}_{h}(1)M_{1}\underline{\psi}_{f}^{\intercal}(1),\cdots ,\underline{\psi}_{h}(z)M_{z}\underline{\psi}_{f}^{\intercal}(z)\right] \Omega_{z_{d}}, \nonumber \\
&=\left[\underline{\psi}_{f}(1)M_{1}\underline{\psi}_{h}^{\intercal}(1),\cdots ,\underline{\psi}_{f}(z)M_{z}\underline{\psi}_{h}^{\intercal}(z)\right] \Omega_{z_{d}},  \\
&= \underline{x}_{f}\Phi_{h}\Omega_{z_{d}}, \nonumber
\end{align}
where, (\ref{eq_middle_step}) is due to the fact that each $\underline{\psi}_{\ell}(i)M_{i}\underline{\psi}_{\ell'}^{\intercal}(i)$ term is a scalar, and $M_{i}$ is a symmetric matrix. The decoder then concatenates all the received repair vectors from helpers in $\mathcal{H}$ to create a $1 \times \alpha$ vector $\underline{\rho}_{\mathcal{H},f}$ as
\begin{align}\label{eq_rho_def}
\underline{\rho}_{\mathcal{H},f} &= [\underline{r}_{h_{1},f}, \cdots , \underline{r}_{h_{d-2b},f}]   \nonumber \\
&= \underline{x}_{f}\left[\Phi_{h_{1}}\Omega_{z_{d}},\cdots,\Phi_{h_{d-2b}}\Omega_{z_{d}}\right].
\end{align}
We define the $\alpha \times \alpha$ matrix $\Theta_{\mathcal{H}}$ as,
\begin{align}\label{eq_Theta}
\Theta_{\mathcal{H}} &= \left[\Phi_{h_{1}}\Omega_{z_{d}},\cdots , \Phi_{h_{d-2b}}\Omega_{z_{d}}\right] \nonumber \\
&= \left[\Phi_{h_{1}},\cdots , \Phi_{h_{d-2b}}\right] \left(I_{(d-2b)} \otimes \Omega_{z_{d}}\right),
\end{align}
where $I_{(d-2b)}$ is the identity matrix of size $(d-2b)$, and $\otimes$ represents the Kronecker product. Note that if matrix $\Theta_{\mathcal{H}}$ is invertible, the decoder will be able to produce an estimate $\underline{\hat{x}}_{\mathcal{H},f}$, for $\underline{x}_{f}$, for any subset $\mathcal{H},~ | \mathcal{H} | = d-2b$, of the test-group, $\mathcal{T}$ as
\begin{align}\label{eq_repair}
\underline{\hat{x}}_{\mathcal{H},f} = \underline{\rho}_{\mathcal{H},f} \Theta_{\mathcal{H}}^{(-1)}.
\end{align}

The following lemma guarantees that $\Theta_{\mathcal{H}}^{(-1)}$ always exists.

\begin{lem}\label{Lem_ThetaRank}
For any $f\in\{1, \cdots , n\}$, any parameter $d \in D$ and a subset of the helper nodes $\mathcal{H}$, with $| \mathcal{H} | = d-2b$, the matrix $\Theta_{\mathcal{H}}$, defined in (\ref{eq_Theta}), is invertible, provided that the code alphabet $\mathbb{F}_{q}$ is large enough.
\end{lem}

For the proof of this lemma please see Appendix \ref{App_ThetaRank}.

A similar discussion as in Lemma \ref{Lem_Consistency} shows that the test-group decoding will always finds a consistent test-group in the data reconstruction and any estimate in a consistent test-group is correct.

To summarize, the test-group decoding for the repair is described in Algorithm \ref{Alg_repair}.
\begin{algorithm}
\caption{Test-group decoding for repair}\label{Alg_repair}
\begin{algorithmic}[1]
\State Input: $d$, $f$, $b$, $\underline{r}_{h,f}$, for all helper node $h$.
\State $\text{Consistency} \gets \text{False}$
\While{$\neg(\text{Consistency})$}
   \State $\mathcal{T}\gets$  A new test-group of size $d-b$
   \For{each subset $\mathcal{H}\subset \mathcal{T}$ with $|\mathcal{H}| = d-2b$}
   	\State Calculate $\underline{\rho}_{\mathcal{H},f}, \Theta_{\mathcal{H},f}$ as defined in (\ref{eq_rho_def}) and (\ref{eq_Theta})
   	\State Calculate $\underline{\hat{x}}_{\mathcal{H},f} \gets \underline{\rho}_{\mathcal{H},f} \Theta^{(-1)}_{\mathcal{H},f} $
   \EndFor
   \If{$\underline{\hat{x}}_{\mathcal{H},f} = \underline{\hat{x}}_{\mathcal{H}',f} ~ \forall{\mathcal{H}, \mathcal{H}' \subset \mathcal{T}}$}
   \State $\text{Consistency} \gets \text{True}$
   \State Output $\gets \underline{\hat{x}}_{\mathcal{H},f}$ for some $\mathcal{H}\subset \text{consistent}~\mathcal{T}$
   \EndIf
\EndWhile
\end{algorithmic}
\end{algorithm}

\begin{rmrk}
The repair procedure presented above requires a large field size. This large field size requirement is mainly due to the specific procedure used for calculating estimates $\underline{\hat{x}}_{\mathcal{H},f}$. We refer to Appendix \ref{App_ThetaRank} for details. In the following section we present an alternative coding scheme that reduces the field size requirement to $|\mathbb{F}_{q} |=n$.
\end{rmrk}

\begin{rmrk}
One may notice that the test-group decoding is a framework that could be used jointly with any coding scheme that provides a mechanism for deriving estimates for the content of the failed node based on the repair data provided by any subset of size $d-2b,~d\in D$. However, it worth mentioning that when $b>0$, the simple concatenation scheme, described in subsection\ref{Subsec_concatenation}, fails to achieve optimal total repair bandwidth when under the test-group decoding framework for the repair. The following example illustrates this fact.
\end{rmrk}

\begin{example}
Let $b = 1$, and consider $k = 3 > 2b$, $D = \{4,5\}$, and $\alpha = 6$ to satisfy (\ref{eq_PernodeCap}). For the simple concatenation scheme to be able to perform repair with $d_{\min}=4$ helpers using the test-group decoding, we need to be able to derive an estimate for the coded content of any failed node based on the data provided by $d_{\min}-2b= 2$ helpers. Therefore, each code component needs to be capable of performing repair with $2$ helpers. Hence, it is clear from the properties of the Product Matrix MBR codes that per node storage capacity of each code component is also $2$. We then conclude that since the overall per node capacity is considered to be $\alpha = 6$ we have 
\begin{align}
z = \frac{\alpha}{d_{\min}-2b} = \frac{6}{2} = 3,  \nonumber
\end{align}
code components.
Now consider a repair based on $d_{\min} = 4$ helpers. As described in Algorithm \ref{Alg_repair}, we need to be able to derive an estimate for the content of the failed node based on every subset of size $d-2b$ helpers, for any $d \in D$. In this case, then we need every subset of $2$ helpers to provide enough repair data to perform a repair in all the $z=3$ code components. We then conclude that every helper should provide a repair symbol for every code component, which results in $\beta(d=4) = 3$, and $\gamma(d=4) = 12$. This matches the per node repair bandwidth of the coding scheme presented in this section as given by (\ref{eq_repairBW_Schm1}), 
\begin{align}
\beta(d) = \frac{\alpha}{d-2b} = \frac{6}{4-2} = 3. \nonumber
\end{align}
However, when we consider the case of $d = 5$, then from (\ref{eq_repairBW_Schm1}) we have $\beta(5) = 2$, and $\gamma(d=5) = 5\beta(d=5) = 10$, in the presented coding scheme. On the other hand, $\gamma(d=5) = 10$ in the simple concatenation scheme, forces at least one of the code components to have less than $4$ helpers. Let the selected set of helpers be denoted by $\{h_{1}, \cdots , h_{5}\}$, and without loss of generality, assume that only $h_{1}, h_{2}$ and $h_{3}$ are providing repair symbols for the code component $1$. Then there exists at least one subset of helpers of size $d-2b = 5-2 = 3$, namely $\mathcal{H} = \{h_{3}, h_{4}, h_{5}\}$, in which only one helper provides repair symbol for code component $1$. Therefore, the test-group decoding is impossible for simple concatenation scheme with $\gamma(d=5) = 10$.
\end{example}

\section{An Alternative Coding Scheme with Small Field Size}\label{Sec_Coding_Sch_II}

The coding scheme presented in this section shares many aspects with the scheme presented in Section \ref{Sec_CodeConstruction}. We use the similar extension of PM MBR codes for encoding the data to be stored in the network. We also use the test-group decoding scheme for data reconstruction and repair, while the procedure for calculating the estimates in the test-group decoding for repair is totally different here. The aim of this alternative solution is to avoid the large field size requirement of the previous scheme. The field size requirements in the previous is imposed by the mechanism of calculating estimates for test-group decoding in the repair, which is based on the non-singularity of matrix $\Theta_{\mathcal{H}}$. In the scheme presented in this section, we provide a different repair procedure which still uses the test-group decoding but works with field size $\mathbb{F}_{q}$ of size $n$.


\subsection{Encoding for Storage and Data Reconstruction Procedure}

Consider the set of parameters $n$, $k$, $D=\{d_{1}, \cdots , d_{\delta}\}$, for an arbitrary flexibility degree $\delta$, $b$, and $\alpha$, such that (\ref{eq_d_i_Cond}) and (\ref{eq_PernodeCap}) are satisfied. In the coding scheme presented in this section the code alphabet $\mathbb{F}_{q}$ only needs to contain $n$ distinct non-zero elements. However, in order to achieve such a small field size, we require the per node storage capacity to satisfy some more constraints as will be discussed later in subsection \ref{Subsec_discission_alpha}. Through this section we continue to use the notation $d_{\min}=\min\{d\in D\}=d_{1}$. 

In order to perform the encoding for storage, we first partition the source data symbols $\underline{s} = [s_{1}, \cdots , s_{F_{\text{MBR}}}]$ into $z=\alpha/(d_{\min}-2b)$ disjoint subsets and arrange them in the form of the overall data matrix $M$, as introduced in (\ref{eq_Information_Matrix}). Next we use the rows of the same $\Psi$ Vandermonde matrix as introduced in (\ref{eq_Psi_def}) as the node-specific coefficient vectors, to encode the data to be stored on each storage node $\ell \in \{1, \cdots , n\}$ as
\begin{align}
\underline{x}_{\ell} = \underline{\psi}_{\ell} M, \nonumber
\end{align}
where, as before, $\underline{\psi}_{\ell}$ denotes the $\ell^{\text{th}}$ row of the matrix $\Psi$. Moreover, for some primitive element $g$ of the code alphabet $\mathbb{F}_{q}$, without loss of generality we consider (\ref{eq_e_i_def}) as the choice for the Vandermone matrix $\Psi$. In other words, we choose
\begin{align}
\Psi = \left[\begin{array}{c c c c c}
1~ & ~g~ & ~g^{2}~ & \cdots & ~g^{\alpha-1} \\
1~ & ~(g^2)~ & ~(g^2)^{2}~ & \cdots & ~(g^2)^{\alpha-1} \\
\vdots & \vdots & \vdots & \ddots & \vdots \\
1~ & ~(g^n)~ & ~(g^n)^{2}~ & \cdots & ~(g^n)^{\alpha-1} \\
\end{array} \right]. \nonumber
\end{align}

Since the encoded consent of the storage nodes is exactly similar to the coding scheme introduced in Section \ref{Subsec_Enc_for_Storage}, we can use the same data reconstruction procedure as presented in Algorithm \ref{Alg_reconstriction}. As a result, the data reconstruction procedure is guaranteed to reconstruct genuine data in the presence of up to $b$ erroneous nodes.

In order to describe the repair procedure, we first need to present some notations and definitions in the following subsection.

\subsection{Notations and Preliminaries for Repair Procedure}

We begin this section by first introducing some new notations and definitions, and a Lemma which will be used to describe the repair scheme. Then we describe the repair based on a single group of helpers of size $d \in D$. Note that, all through the procedure we always require all the helpers to perform similar procedures on their content and hence the provided repair symbols are always symmetric. As a result, the choice of helpers do not change the procedure and the same scheme could be performed based on any other subset of helpers of the same size. The only determinant parameter is the size of the group of helpers to be used for generating a single estimate.

Let
\begin{align}\label{eq_xi_def}
\xi = \left\lfloor \frac{(d-2b)}{(d_{\min}-2b)}\right\rfloor (d_{\min}-2b),
\end{align}
which results in $\xi \leq (d-2b) < 2 \xi$. Also let
\begin{align}
\zeta = \frac{\alpha}{\xi}. \nonumber
\end{align}
We will assume $\zeta$ is an integer for any choice of $d \in D$.

One can consider the data matrix $M$, as follows,
\begin{align}
M = \left[\begin{array}{c c c c c}
M'_{1} & O & O & \cdots & O \\
O & M'_{2} & O & \cdots & O \\
\vdots & ~~ & \ddots & ~~ & \vdots \\
O & O & \cdots & O & M'_{\zeta}
\end{array}\right], \nonumber
\end{align}
where, $M'_{i},~i\in\{1,\cdots , \zeta\}$ is an $\xi \times \xi$ block diagonal matrix with $c=\frac{\xi}{d_{\min}-2b}$ diagonal blocks, $M_{(i-1)c+1}$, $M_{(i-1)c+2}$, $\cdots$ , $M_{ic}$. This structure is depicted in Fig \ref{Fig_M_primes}. 

\begin{figure}
\centering
\resizebox{3.25 in}{!}{\includegraphics{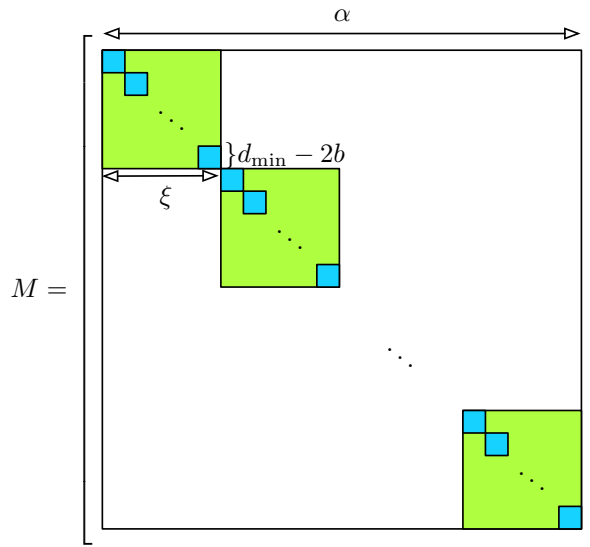}}
\caption{The structure of the data matrix $M$, with its submatrices $M'_{i},~i\in\{1,\cdots , \zeta\}$, depicted as diagonal green squares. Each $\xi \times \xi$ submatrix $M'_{i}$ is itself a block diagonal matrix with $(d_{\min}-2b)\times (d_{\min}-2b)$ diagonal blocks, $M_{j},~j\in\{(i-1)c+1 , \cdots , ic\}$ depicted as small blue squares.}\label{Fig_M_primes}
\end{figure}

Accordingly, for each node, $\ell$, we consider a partitioning on its node-specific coefficient vector, $\underline{\psi}_{\ell}$, as well as its coded content, $\underline{x}_{\ell}$, to disjoint consequent segments of size $\xi$ as follows
\begin{align}
\underline{\psi}_{\ell} = [\underline{\varphi}_{\ell}(1), \underline{\varphi}_{\ell}(2), \cdots , \underline{\varphi}_{\ell}(\zeta)], \nonumber \\
\underline{x}_{\ell} = [\underline{\chi}_{\ell}(1), \underline{\chi}_{\ell}(2), \cdots , \underline{\chi}_{\ell}(\zeta)]. \nonumber
\end{align}
Therefore, each segment $\underline{\chi}_{\ell}(i)$ is associated with the submatrix $M'_{i}$ of the data matrix as
\begin{align}\label{eq_x_ell_segment}
\underline{\chi}_{\ell}(i) = \underline{\varphi}_{\ell}(i)M'_{i},~~i \in\{1, \cdots , \zeta \}.
\end{align}

We now present the following definition and its following lemma, which would be used through the repair scheme.

\begin{defn}[Merge Operator]
Consider an element $e\in \mathbb{F}_{q}$, and vectors $\underline{v},\underline{u}\in \mathbb{F}_{q}^{\xi}$, and denote the $i^{\text{th}}$ element of $\underline{v}$ and $\underline{u}$ by $v_{i}$ and $u_{i}$ respectively. For integers $m, \epsilon$, such that  $2\xi > m \geq \xi$, and $\epsilon > 1$, the merge operator $\Phi_{m,\xi}: \mathbb{F}_{q} \times \mathbb{F}_{q}^{\xi} \times \mathbb{F}_{q}^{\xi}\rightarrow\mathbb{F}_{q}^{m}$, is defined as
\begin{align}
\Phi_{m,\epsilon}(e,\underline{v},\underline{u}) = \underline{\phi} = [\phi_{1}, \phi_{2}, \cdots , \phi_{m}], \nonumber
\end{align}
where,
\begin{align}\label{eq_merge_def}
\phi_{i} = \left\lbrace \begin{array}{c c}
v_{i} ~ & ~ \text{for}~i \in\{1, \cdots , m-\xi \}, \\
v_{i}+e^{m-\epsilon\xi}u_{i+\xi-m} ~ & ~ \text{for}~i \in\{m-\xi+1, \cdots , \xi \}, \\
e^{m-\epsilon\xi}u_{i+\xi-m} ~ & ~ \text{for}~i \in\{\xi+1, \cdots , m \}.
\end{array}  \right.
\end{align}
In other words,
\begin{align}\label{eq_merge}
\Phi_{m,\epsilon}(e,\underline{v},\underline{u}) = [\underline{v}, \underbrace{0, \cdots, 0}_{m-\xi}] + e^{m-\epsilon\xi}[\underbrace{0, \cdots, 0}_{m-\xi}, \underline{u}].
\end{align}
\end{defn} 

Figure \ref{Fig_merge_def} depicts the definition of the merge operator as defined above.

\begin{figure}
\centering
\resizebox{3 in}{!}{\includegraphics{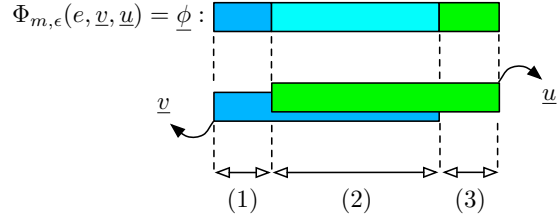}}
\caption{The result of performing the merge operator $\Phi_{m,\epsilon}$ on vectors $\underline{v}$, and $\underline{u}$. Segment $(1)$ is of length $m-\xi$ and consists only of entries of vector $\underline{v}$. Segment $(2)$ is of length $2\xi -m$ and its entries are combinations of entries of $\underline{v}$ and $\underline{u}$. Segment $(3)$ is also of length $m-\xi$ and its entries come from $\underline{u}$, scaled by a constant.}\label{Fig_merge_def}
\end{figure}

The following lemma presents an important observation, connecting the outcome of the merge operator performed on two segments of the coded content in a helper node and its counterpart in the failed node. In order to present the lemma, we need the following notation. For each node $\ell$ in the network, let $e_{\ell}$ be the element in the code alphabet $\mathbb{F}_{q}$ associated with the coefficient vector assigned to node $\ell$, and consider an integer $m$, such that $\xi \leq m < 2\xi$. We will use the notation $\underline{\psi}_{\ell,m}$ for the following $1\times m$ vector,
\begin{align}
\underline{\psi}_{\ell,m} = [e_{\ell}^{0}, e_{\ell}^{1}, \cdots , e_{\ell}^{m-1}]. \nonumber
\end{align}
Note that $\underline{\psi}_{\ell,m}$ indeed, denotes the first $m$ entries of the node-specific coefficient vector $\underline{\psi}_{\ell}$

\begin{lem}\label{Lem_merge}
Consider the elements $e_{h}$, and $e_{f}$, in the code alphabet $\mathbb{F}_{q}$, associated with the coefficient vectors of a helper node $h$ and the failed node $f$ respectively. For an integers $m$ such that $\xi \leq m  < 2\xi$, and two integers $i,j\in\{1, \cdots , \zeta\}$  with $i<j$, let $\epsilon = j-i+1$. Then we have,
\begin{align}
e_{h}^{(i-1)\xi}\left(\Phi_{m,\epsilon}(e_{f},\underline{\chi}_{f}(i), \underline{\chi}_{f}(j)) \underline{\psi}_{h,m}^{\intercal}\right) = e_{f}^{(i-1)\xi}\left(\Phi_{m,\epsilon}(e_{h},\underline{\chi}_{h}(i), \underline{\chi}_{h}(j)) \underline{\psi}_{f,m}^{\intercal}\right). \nonumber
\end{align}
\end{lem}

\begin{proof}
From (\ref{eq_merge}) it is easy to see that 
\begin{align}\label{eq_Phi_f}
\Phi_{m,\epsilon}(e_{f},\underline{\chi}_{f}(i), \underline{\chi}_{f}(j)) &= [\underline{\chi}_{f}(i), \underbrace{0, \cdots, 0}_{m-\xi}] + e_{f}^{m-\epsilon\xi}[\underbrace{0, \cdots, 0}_{m-\xi}, \underline{\chi}_{f}(j)].
\end{align}
Now, using (\ref{eq_x_ell_segment}) we can rewrite the two terms on the right hand side above as,
\begin{align}\label{eq_first_term}
[\underline{\chi}_{f}(i), \underbrace{0, \cdots, 0}_{m-\xi}] &= e_{f}^{(i-1)\xi}\underline{\psi}_{f,m}\left[ \begin{array}{ c c }
M'_{i} ~ & ~ O_{1} \\
O_{1}^{\intercal} ~ & ~ O_{2}
\end{array} \right], 
\end{align}
and
\begin{align}\label{eq_second_term}
e^{m-\epsilon\xi}[\underbrace{0, \cdots, 0}_{m-\xi}, \underline{\chi}_{f}(j)] &= e_{f}^{m-(j-i+1)\xi}\left( e_{f}^{(j-m)\xi}\underline{\psi}_{f,m}\left[ \begin{array}{ c c }
O_{2} ~ & ~ O_{1}^{\intercal} \\
O_{1} ~ & ~ M'_{j}
\end{array} \right]\right) \nonumber  \\
&= e_{f}^{(i-1)\xi}\underline{\psi}_{f,m}\left[ \begin{array}{ c c }
O_{2} ~ & ~ O_{1}^{\intercal} \\
O_{1} ~ & ~ M'_{j}
\end{array} \right].
\end{align}
where, $O_{1}$ and $O_{2}$ are all zero matrices of size $\xi \times (m-\xi)$ and $(m-\xi) \times (m-\xi)$ respectively. Let $\Lambda_{i,j}$ denote the $m \times m$ symmetric matrix defined as,
\begin{align}
\Lambda_{i,j} = \left[ \begin{array}{ c c }
M'_{i} ~ & ~ O_{1} \\
O_{1}^{\intercal} ~ & ~ O_{2}
\end{array} \right] + \left[ \begin{array}{ c c }
O_{2} ~ & ~ O_{1}^{\intercal} \\
O_{1} ~ & ~ M'_{j}
\end{array} \right]. \nonumber
\end{align}
Therefore, from (\ref{eq_Phi_f}), (\ref{eq_first_term}), and (\ref{eq_second_term}) we have,
\begin{align}\label{eq_Phi_f_rewrite}
\Phi_{m,\epsilon}(e_{f},\underline{\chi}_{f}(i), \underline{\chi}_{f}(j)) = e_{f}^{(i-1)\xi}\underline{\psi}_{f,m} \Lambda_{i,j}.
\end{align}
Similarly, one can show that 
\begin{align}\label{eq_Phi_h_rewrite}
\Phi_{m,\epsilon}(e_{h},\underline{\chi}_{h}(i), \underline{\chi}_{h}(j)) = e_{h}^{(i-1)\xi}\underline{\psi}_{h,m} \Lambda_{i,j}.
\end{align}
Then we have,
\begin{align}\label{eq_first_side}
e_{h}^{(i-1)\xi}\left(\Phi_{m,\epsilon}(e_{f},\underline{\chi}_{f}(i), \underline{\chi}_{f}(j)) \underline{\psi}_{h,m}^{\intercal}\right) &= e_{h}^{(i-1)\xi} \left(\left( e_{f}^{(i-1)\xi} \underline{\psi}_{f,m} \Lambda_{i,j} \right) \underline{\psi}_{h,m}^{\intercal} \right) \nonumber \\
&= e_{h}^{(i-1)\xi} e_{f}^{(i-1)\xi} \left(\underline{\psi}_{f,m} \Lambda_{i,j} \underline{\psi}_{h,m}^{\intercal} \right)^{\intercal} \nonumber \\
&= e_{f}^{(i-1)\xi} e_{h}^{(i-1)\xi}  \left(\underline{\psi}_{h,m} \Lambda_{i,j} \underline{\psi}_{f,m}^{\intercal} \right). 
\end{align}
In the above, (\ref{eq_first_side}) follows from the fact that $\underline{\psi}_{f,m} \Lambda_{i,j} \underline{\psi}_{h,m}^{\intercal}$ is a scalar. Hence, using (\ref{eq_second_term}) we have,
\begin{align}
e_{f}^{(i-1)\xi} e_{h}^{(i-1)\xi}  \left(\underline{\psi}_{h,m} \Lambda_{i,j} \underline{\psi}_{f,m}^{\intercal} \right) = e_{f}^{(i-1)\xi}\left(\Phi_{m,\epsilon}(e_{h},\underline{\chi}_{h}(i), \underline{\chi}_{h}(j))  \underline{\psi}_{f,m}^{\intercal}\right), \nonumber
\end{align}
which completes the proof.
\end{proof}

Based on the Lemma \ref{Lem_merge}, we now introduce \emph{"$m$-merged repair symbol from segments $i$ and $j$"} in the following definition. 

\begin{defn}\label{Def_m_merged_repair_symb}[$m$-merged Repair Symbol]
For the fixed $\xi$, and any integer $m,~\xi \leq m < 2\xi$, each helper $h$ can merge two segments $\underline{\chi}_{h}(i)$, and $\underline{\chi}_{h}(j)$ of its coded content by setting $\epsilon = j-i+1$ and create a repair symbol as
\begin{align}\label{eq_merge_repair_symb}
\rho_{m,h,f}(i,j) = e_{f}^{(i-1)\xi}\left(\Phi_{m,\epsilon}(e_{h},\underline{\chi}_{h}(i), \underline{\chi}_{h}(j))  \underline{\psi}_{f,m}^{\intercal}\right).
\end{align}
We will refer to this repair symbol as the  \emph{"$m$-merged repair symbol from segments $i$ and $j$"}.
\end{defn}

Lemma \ref{Lem_merge} then guarantees that such a repair symbols provides a linear equation in terms of the symbols resulting from merging segments $\underline{\chi}_{f}(i)$, and $\underline{\chi}_{f}(j)$. The coefficient of this linear equation is given by
\begin{align}
e_{h}^{(i-1)\xi}\underline{\psi}_{h,m}. \nonumber
\end{align}

\subsection{Repair Scheme}

Now we have everything ready to describe the repair scheme. In order to provide the required error resiliency, we use the test-group decoder scheme as described in the previous Section. However, here we use different encoding and estimate calculation procedures. To this end, for a set of $d$ helpers, $\binom{d}{d-2b}$ parallel decoding procedures calculate all possible estimates for $\underline{\hat{x}}_{f,\mathcal{H}}$, for any subset $\mathcal{H}$ of size $d-2b$ of the selected helpers. The decoder then decides the correct decoding result by checking the consistency among all estimates derived from subsets of any test-group $\mathcal{T}$ with $|\mathcal{T}| = d-b$, as described in Algorithm \ref{Alg_repair}.

In the rest of this subsection we will describe the repair scheme, based on a set of $d$ helpers, in the form of an iterative process. In each iteration we have an encoding step which is performed similarly by all the participating helpers, and produces some repair symbols. We also have a decoding step in each iteration which is performed at the repair decoder. The goal of the decoding procedure is to calculate an estimate for the coded content of the failed node based on every subset of size $d-2b$ of the selected helpers. Every encoding step is the same for all the helpers and in every decoding step the procedure to be performed for calculating the estimate based on all subsets is similar. Therefore in the rest of this discussion, we only focus of a single arbitrary subset of size $d-2b$. 

Let $f$ denote the index of the failed node, and denote the coded content of the failed node as
\begin{align}
\underline{x}_{f} = [\underline{\chi}_{f}(1), \cdots , \underline{\chi}_{f}(\zeta)], \nonumber
\end{align}
where each $\underline{\chi}_{f}(i)$ is referred to as a segment, and contains $\xi$ entries. In the process of calculating an estimate $\hat{\underline{x}}_{f}$, the decoder calculates an estimate for every single entry of every single segment. The number of iterations for a repair procedure depends on the selected parameter $d$, and its corresponding $\xi$ as defined in (\ref{eq_xi_def}). Recall that from the definition of $\xi$ for a given $d \in D$ we have
\begin{align}
\xi \leq d-2b < 2\xi \nonumber
\end{align}
If we have $\xi = d-2b$, then the number of entries to be estimated in each segment matches the number of the helpers from which the test-group decoder needs to calculate an estimate. In this case the repair scheme consists of only one iteration as follows. Each helper $h$ produces one repair symbol from each segment of its coded content, $\underline{x}_{h}$, as 
\begin{align}
r_{h,i} = \underline{\chi}_{h}(i)\left(\underline{\varphi}_{f}(i)\right)^{\intercal},~~i \in \{1, \cdots, \zeta\} \nonumber
\end{align}
and sends the repair symbols to the repair decoder. Note that for $i \in \{1, \cdots, \zeta\}$ then we have,
\begin{align}\label{eq_repair_symbol_first_iteration}
r_{h,i} &= \underline{\chi}_{h}(i)\left(\underline{\varphi}_{f}(i)\right)^{\intercal} \nonumber \\
&= \underline{\varphi}_{h}(i)M'_{i}\left(\underline{\varphi}_{f}(i)\right)^{\intercal} \nonumber \\
&= \underline{\varphi}_{f}(i)M'_{i}\left(\underline{\varphi}_{h}(i)\right)^{\intercal} \nonumber \\
&= \underline{\chi}_{f}(i)\left(\underline{\varphi}_{h}(i)\right)^{\intercal}
\end{align}

At the repair decoder, fro each segment $i$, stacking all the repair symbols $r_{h_{\ell},i}$ received from each subset of helpers $\mathcal{H} = \{h_{1}, \cdots , h_{d-2b}\}$, from (\ref{eq_repair_symbol_first_iteration}) we have,
\begin{align}\label{eq_single_iteration_repair_system}
[r_{h_{1},i}, \cdots , r_{h_{d-2b},i}] = \underline{\chi}_{f}(i)\left[\left(\underline{\varphi}_{h_{1}}(i)\right)^{\intercal}, \cdots , \left(\underline{\varphi}_{h_{d-2b}}(i)\right)^{\intercal}\right],~~i \in \{1, \cdots, \zeta\}
\end{align}
It is easy to check that the coefficient matrix on the right hand side of the above equation is $(d-2b)\times (d-2b)$, and it is invertible, since it could be decomposed to a diagonal full ranks matrix multiplied into a Vandermonde matrix. Therefore, the decoder is able to calculate every segment $\underline{\chi}_{f}(i)$ based on the repair symbols provided by the helpers in $\mathcal{H} = \{h_{1}, \cdots , h_{d-2b}\}$, and the repair scheme terminates at the end of this single iteration.

On the other hand, when we have $\xi < d-2b < 2\xi$, the number of entries in each segment $\underline{\chi}_{f}(i)$ is larger than $d-2b$, and hence the linear equations system introduced in (\ref{eq_single_iteration_repair_system}) is not uniquely solvable. As a result, the single iteration introduced above is not applicable. In this case the repair scheme has more than one iteration. In the rest of this subsection we describe these iterations. We use a numerical example, evolving through the description of procedures, to better illustrate the procedures. 

In the first iteration every helper forms disjoint groups consisting of two consecutive segments of its coded content, and uses the merge operator $\Phi_{(d-2b),1}$ introduced in the previous subsection, to merge the segments in each group. Let $I_{g}=\{2g-1,2g\}$, denote the set of indices of the segment in group $g$, then we the result of the merge operator in group $g$ at helper $h$ in iteration one is,
\begin{align}\label{eq_phi_h_1_g}
\underline{\phi}_{h,1,g} = \Phi_{(d-2b), 1}(e_{h}, \underline{\chi}_{h}(2g-1), \underline{\chi}_{h}(2g)).
\end{align}
Finally each helper $h$ creates one $(d-2b)$-merged repair symbol from each group of two segments, as defined in Definition \ref{Def_m_merged_repair_symb}. For instance, the repair symbol from helper $h$ based on the group $g$ in iteration one, is created as
\begin{align}
r_{h,1,g} = e_{f}^{(2g-2)\xi}\underline{\phi}_{h,j,g}\underline{\psi}_{f,(d-2b)}^{\intercal} = \rho_{(d-2b),h,f}(2g-1,2g). \nonumber
\end{align}
Similar to (\ref{eq_phi_h_1_g}), let us denote 
\begin{align}
\underline{\phi}_{f,1,g} = \Phi_{(d-2b), 1}(e_{f}, \underline{\chi}_{f}(2g-1), \underline{\chi}_{f}(2g)). \nonumber
\end{align}
Hence, using Lemma \ref{Lem_merge}, we have
\begin{align}
r_{h,1,g} = \rho_{(d-2b),h,f}(2g-1,2g) = \underline{\phi}_{f,1,g} \left(e_{h}^{(2g-2)\xi}\underline{\psi}_{h,(d-2b)}^{\intercal}\right). \nonumber
\end{align}

Then, the repair symbols provided by any subset, $\mathcal{H}=\{h_{1}, \cdots , h_{(d-2b)}\}$, of size $(d-2b)$ of the helpers provides a linear equation system in terms of the entries in segments $\underline{\chi}_{f}(2g-1)$ and $\underline{\chi}_{f}(2g)$, for any group $g$, as follows,
\begin{align}
\left[r_{h_{1},1,g}, \cdots , r_{h_{(d-2b)},1,g} \right] = \underline{\phi}_{f,1,g} \left[\underline{\psi}_{h_{1},(d-2b)}^{\intercal}, \cdots , \underline{\psi}_{h_{(d-2b)},(d-2b)}^{\intercal}\right]. \nonumber
\end{align}
In the above, the columns of the coefficient matrix on the right hand side are linearly independent, and the matrix is indeed a $(d-2b)\times (d-2b)$ Vandermonde matrix, which is invertible. Therefore the repair decoder is able to calculate all the entries in $\underline{\phi}_{f,1,g}$, for every group $g$, based on the repair symbols provided by each subset of helpers of size $d-2b$ in the first iteration. However, notice that the entries in $\underline{\phi}_{f,1,g}$ can be categorized into two categories as depicted in Fig. \ref{Fig_Merged_segment_iteration_one}, where category (1) contains $2\xi - (d-2b)$ entries and category (2) consists of $2(d-2b-\xi)$ entries in total. 

\begin{figure}
\centering
\resizebox{2.5 in}{!}{\includegraphics{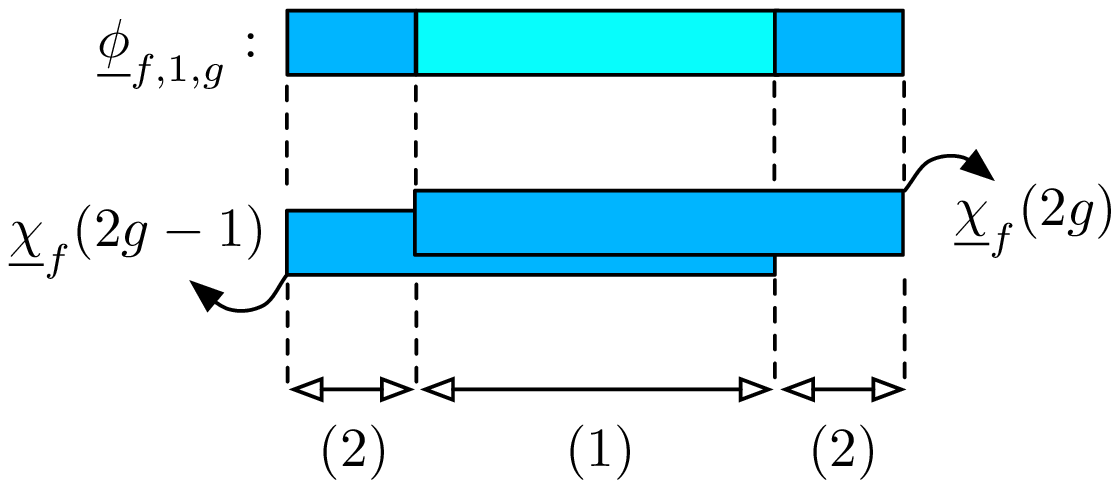}}
\caption{An active segment $\underline{\chi}(i)$; all active entries (depicted in blue) have indices less than known entries (depicted in green).}\label{Fig_Merged_segment_iteration_one}
\end{figure}

The entries in category (2) are either the same as a single entry from $\underline{\chi}_{f}(2g-1)$, or a scaled version of a single entry in $\underline{\chi}_{f}(2g)$, where the scaling factor is $e_{f}^{d-2b-\xi}$. Then calculating each of the entries in category (2) of $\underline{\phi}_{f,1,g}$ recovers the value of one entry from the lost coded vector $\underline{x}_{f}$. On the other hand, each of the entries in category (1) of $\underline{\phi}_{f,1,g}$ is formed by a linear combination of one entry from $\underline{\chi}_{f}(2g-1)$ and one entry from $\underline{\chi}_{f}(2g)$. Therefore, by recovering the value of each entry in category (1) of $\underline{\phi}_{f,1,g}$, the decoder is only able to calculate one entry from $\underline{\chi}_{f}(2g-1)$ in terms of an entry from $\underline{\chi}_{f}(2g)$.

Note that in the process of calculating an estimate of $\underline{x}_{f}$, the decoder needs to calculates an estimate for every single entry of every single segment. At the end of the iteration one, as explained above, the decoder calculates some estimates for some of the entries, some other entries are estimated in terms of some other entries, and for the rest of the entries we have not calculated any estimate yet. In order to keep track of this process, we assign a label to each entry of $\underline{x}_{f}$. We set the label for all entries which are already estimated as \emph{"known"}. At the end of iteration one, known entries include all entries of $\underline{x}_{f}$ corresponding to an entry in category (2) of $\underline{\phi}_{f,1,g}$, for all $g$.

The entries of $\underline{x}_{f}$ which are not known yet at the end of iteration one, are the entries which have been combined to form an entry of category (1) in $\underline{\phi}_{f,1,g}$, for some $g$. As described above, among these entries, we can estimate each entry from $\underline{\chi}_{f}(2g-1)$ in terms of another entry in $\underline{\chi}_{f}(2g)$ at the end of iteration one. We label all such entries in $\underline{\chi}_{f}(2g-1)$ as \emph{"inactive"}. The decoder does not need to work on calculating the estimate of the inactive entries any more, since their explicit estimate will be evaluated once all the other entries are estimated.

All the other entries then need to be estimated yet in the following iterations and hence we label them as "active". We refer to an entry with "active" label as an active entry. If the decoder calculates an explicit estimate for an active entry in some later iteration, the label for that entry changes to "known". Similarly, the decoder may calculate an estimate for it in terms of another active entry and hence change its label to "inactive".

At each iteration we also refer to a segment $\underline{\chi}_{f}(i)$ as active if it contains at least one active entry. The repair procedure terminates when there is no active entry left. We will show that both the number of active segments as well as the number of active entries in each active segment reduces as we move through the steps of the decoding.

It is easy to check all the following properties are satisfied at the end of iteration one.
\begin{itemize}
\item The number of active entries in each active segment will always be the same for all active segments in each iteration.
\item Any active segment will only contain either active or known entries. In other words, an active segment will never contain an inactive entry.
\item In each active segment the indices of the active entries are always less than the indices of the known entries (see Fig. \ref{Fig_active_segment}).
\end{itemize}

The decoder will then proceed through the next iterations by keeping all these properties as invariants, as will be described in the following.

\begin{figure}
\centering
\resizebox{2.5 in}{!}{\includegraphics{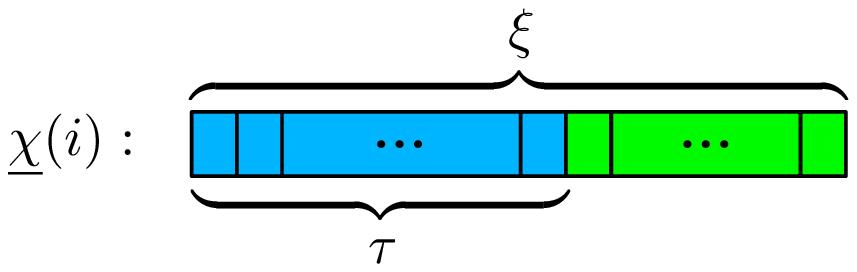}}
\caption{An active segment $\underline{\chi}(i)$; all active entries (depicted in blue) have indices less than known entries (depicted in green).}\label{Fig_active_segment}
\end{figure}

Let $\tau_{j}$ denote the number of active entries in each active segment at the beginning of iteration $j$. Since all entries in all segments are active at the beginning of the repair procedure, then we have
\begin{align}
\tau_{1} = \xi \nonumber
\end{align}
Also let non-negative integers $\mu_{j}$ and $\sigma_{j}$ be such that
\begin{align}\label{eq_mu_sigma_def}
(d-2b) = \mu_{j}\tau_{j} + \sigma_{j},~~0\leq \sigma_{j} < \mu_{j}.
\end{align}

\begin{example}\label{Ex_2nd_repair}
Consider the parameters $\delta = 2$, $n=6$, $k=3$, $d_{\min} = d_{1} = 4$, $d_{\delta} = d_{2} = 5$, $b = 1$, and assume $\alpha = 12$. The content of each node $i$ consists of $\underline{x}_{i} = [x_{i,1} , \cdots , \underline{x}_{i,12}]$, and at the beginning all 12 entries are active. We will consider the case $d =5$, and assume node $f=6$ is failed, helpers are $\{1, \cdots , 5\}$, and we focus on the procedure of calculating an estimate based on the subset $\mathcal{H}=\{1, 2, 3\}$. From (\ref{eq_xi_def}) we have $\xi = 2$. Moreover, we have,
\begin{align}
\underline{\chi}_{\ell}(1) &= [x_{\ell,1}, x_{\ell,2}], ~ \underline{\chi}_{\ell}(2) = [x_{\ell,3}, x_{\ell,4}],~~\underline{\chi}_{\ell}(3) = [x_{\ell,5}, x_{\ell,6}], \nonumber \\
\underline{\chi}_{\ell}(4) &= [x_{\ell,7}, x_{\ell,8}], ~ \underline{\chi}_{\ell}(5) = [x_{\ell,9}, x_{\ell,10}],~\underline{\chi}_{\ell}(6) = [x_{\ell,11}, x_{\ell,12}]. \nonumber
\end{align}

In step $j = 1$ then $\tau_{1} = \xi = 2$ and from (\ref{eq_mu_sigma_def}) we have, $\mu_{1} = 1$, and $\sigma_{1} = 1$.
\end{example}

The rest of the repair procedure then depend on whether $\sigma_{j} > 0$ or $\sigma_{j} = 0$. In the following we will first assume $\sigma_{j} > 0$, and then at the end of this subsection we describe the case of $\sigma_{j}=0$.

\subsubsection{Case of $\sigma_{j}>0$}

When $\sigma_{j} > 0$, we group each $\mu_{j}+1$ consequent active segments. Each helper then modifies each group by merging the last two active segments in each group.  Let 
\begin{align}
I_{g} = \{i_{1}, i_{2}, \cdots , i_{\mu_{j}}, i_{\mu_{j}+1}\},~\text{and},~i_{1}<i_{2}<\cdots <i_{\mu_{j}}<i_{\mu_{j}+1}. 
\end{align}
denote the set of indices of the active segments in the $g^{\text{th}}$ group. Each helper $h$ then use the merge operator $\Phi_{m_{j},\epsilon_{j,g}}$ for 
\begin{align}\label{eq_m_ell_sigma_ell}
m_{j} = \xi + \sigma_{j}, 
\end{align}
and 
\begin{align}\label{eq_epsilon_ell_j}
\epsilon_{j,g} = i_{\mu_{j}+1} - i_{\mu_{j}} + 1, 
\end{align}
to merge the last two segments in the group as follows
\begin{align}\label{eq_phi_h_ell_j}
\underline{\phi}_{h,j,g} = \Phi_{m_{j}, \epsilon_{j,g}}(e_{h}, \underline{\chi}_{h}(i_{\mu_{j}}), \underline{\chi}_{h}(i_{\mu_{j}+1})).
\end{align}

\begin{example}
Let's continue Example \ref{Ex_2nd_repair}. In step 1, there are 3 groups,
\begin{align}
I_{1} = \{1,2\},~I_{2} = \{3,4\},~I_{3} = \{5,6\}. \nonumber
\end{align}
Moreover from (\ref{eq_m_ell_sigma_ell}) we have,
\begin{align}
m_{1} = \xi+\sigma_{1} = 3,~~\epsilon_{1,1}= \epsilon_{1,2}=\epsilon_{1,3}=1. \nonumber
\end{align}
Then each helper $h$ merges,
\begin{align}
\underline{\phi}_{h,1,1} = \Phi_{3,1}(e_{h},[x_{h,1},x_{h,2}],[x_{h,3},x_{h,4}]) = \left[x_{h,1}, (x_{h,2}+e_{h}^{-1}x_{h,3}), (e_{h}^{-1}x_{h,4}) \right], \nonumber
\end{align}
and similarly,
\begin{align}
\underline{\phi}_{h,1,2} &= \left[x_{h,5}, (x_{h,6}+e_{h}^{-1}x_{h,7}), (e_{h}^{-1}x_{h,8}) \right], \nonumber \\
\underline{\phi}_{h,1,3} &= \left[x_{h,9}, (x_{h,10}+e_{h}^{-1}x_{h,11}), (e_{h}^{-1}x_{h,12}) \right]. \nonumber
\end{align}
\end{example}

Finally each helper $h$ creates one repair symbol from each modified group of active segments. For instance in iteration $j$, the repair symbol from helper $h$ based on the active segments in group $g$ is created as
\begin{align}\label{eq_repair_symbol_def}
r_{h,j,g} &= e_{f}^{(i_{\mu_{j}}-1)\xi}\underline{\phi}_{h,j,g}\underline{\psi}_{f,m_{j}}^{\intercal} + \sum_{i \in I_{g}\setminus \{i_{\mu_{j}},i_{\mu_{j}+1}\}}{\underline{\chi}_{h}(i)\underline{\varphi}_{f}^{\intercal}(i)}. \nonumber \\
&= \rho_{m_{j},h,f}(i_{\mu_{j}},i_{\mu_{j}+1}) + \sum_{i \in I_{g}\setminus \{i_{\mu_{j}},i_{\mu_{j}+1}\}}{\underline{\chi}_{h}(i)\underline{\varphi}_{f}^{\intercal}(i)}. 
\end{align}

Note that for each helper $h$, and each active segment indexed $i$ we have
\begin{align}\label{eq_segment_i_repair}
\underline{\chi}_{h}(i)\underline{\varphi}_{f}^{\intercal}(i) &= \underline{\varphi}_{h}(i)M'_{i}\underline{\varphi}_{f}^{\intercal}(i) \nonumber \\
&= \underline{\varphi}_{f}(i)M'_{i}\underline{\varphi}_{h}^{\intercal}(i) \nonumber \\
&= \underline{\chi}_{f}(i)\underline{\varphi}_{h}^{\intercal}(i). 
\end{align}
Moreover, using Lemma \ref{Lem_merge}, we have
\begin{align}\label{eq_segment_i_merge_repair}
\rho_{m_{j},h,f}(i_{\mu_{j}},i_{\mu_{j}+1}) = \underline{\phi}_{f,j,g}\left(e_{h}^{(i_{\mu_{j}}-1)\xi}\underline{\psi}_{h,m_{j}}^{\intercal}\right).
\end{align}
Therefore, from (\ref{eq_segment_i_repair}) and (\ref{eq_segment_i_merge_repair}) it is clear that 
\begin{align}\label{eq_r_h_def}
r_{h,j,g} = [\underline{\chi}_{f}(i_{1}), \cdots , \underline{\chi}_{f}(i_{\mu_{j}}-1),\underline{\phi}_{f,j,g}] \underline{\vartheta}_{h,g}^{\intercal},
\end{align}
where,
\begin{align}\label{eq_theta_def}
\underline{\vartheta}_{h,g} = [\underline{\varphi}_{h}(i_{1}), \cdots , \underline{\varphi}_{h}(i_{\mu_{j}}-1), e_{h}^{(i_{\mu_{j}}-1)\xi}\underline{\psi}_{h,m_{j}}].
\end{align}

In (\ref{eq_r_h_def}), each active segment $\underline{\chi}_{f}(i),~i\in I_{g}\setminus \{i_{\mu_{j}}, i_{\mu_{j}+1}\}$ consists of $\tau_{j}$ active entries and $\xi - \tau_{j}$ known entries. Moreover, the entries of $\underline{\phi}_{f,j,g}$ could be divided into three categories, as depicted in Fig. \ref{Fig_merge_segment}; (1) Entries formed by combining two active entries in $\underline{\chi}_{f}(i_{\mu_{j}})$, and $\underline{\chi}_{f}(i_{\mu_{j}+1})$, depicted in cyan color, (2) Entries formed either from one active entry or from combining an active entry by a known entry, depicted in blue color, and (3) Known entries depicted in green color.

\begin{figure}
\centering
\resizebox{3 in}{!}{\includegraphics{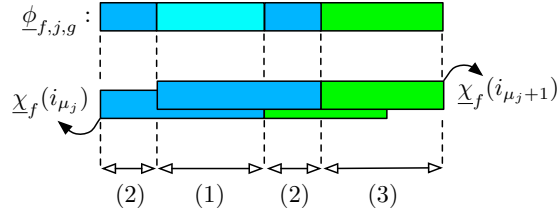}}
\caption{The three categories of entries in $\underline{\phi}_{f,j,g}$.}\label{Fig_merge_segment}
\end{figure}

Hence, the repair symbols provided by any subset, $\mathcal{H}=\{h_{1}, \cdots , h_{(d-2b)}\}$, of size $(d-2b)$ of the helpers provides a linear equation system in terms of the entries in segments $\underline{\chi}_{f}(i),~i \in I_{g}$, for any group $g$, as follows,
\begin{align}\label{eq_lin_eq_sys}
\left[r_{h_{1},j,g}, \cdots , r_{h_{(d-2b)},j,g} \right] = [\underline{\chi}_{f}(i_{1}), \cdots , \underline{\chi}_{f}(i_{\mu_{j}}-1),\underline{\phi}_{f,j,g}] \Theta_{g}.
\end{align}
In the above, the columns of the coefficient matrix for group $g$, namely $\Theta_{g}$, are,
\begin{align}
\Theta_{g} = [\underline{\vartheta}_{h_{1},g}^{\intercal}, \cdots , \underline{\vartheta}_{h_{(d-2b)},g}^{\intercal}]. \nonumber
\end{align}

\begin{example}\label{Ex_2nd_repair_step_1}
Following the setting considered in Example \ref{Ex_2nd_repair}, in iteration one of the repair scheme, we have three groups; $I_{1}=\{1,2\}$, $I_{2}=\{3,4\}$, $I_{3}=\{5,6\}$. Moreover, as described in last examples we have $m_{1} = 3$, $\xi = 2$, and $\mu_{1}=\epsilon_{1,1}=\epsilon_{1,2}=\epsilon_{1,3}=1$. From the first group, using (\ref{eq_repair_symbol_def}), each helper $h$ then provides the repair symbol
\begin{align}
r_{h,1,1} = e_{f}^{0}\underline{\phi}_{h,1,1}\underline{\psi}_{f,m_{1}}^{\intercal}=\left[x_{h,1}, (x_{h,2}+e_{h}^{-1}x_{h,3}, e_{h}^{-1}x_{h,4}) \right]\left[1, e_{f}, e_{f}^{2} \right]^{\intercal}. \nonumber
\end{align}
According to Lemma \ref{Lem_merge}, we have
\begin{align}
r_{h,1,1} =\left[x_{f,1}, (x_{f,2}+e_{f}^{-1}x_{f,3}), e_{f}^{-1}x_{f,4} \right]\left[1, e_{h}, e_{h}^{2} \right]^{\intercal}. \nonumber
\end{align}
Therefore, using $r_{h_{1},1,1}$, $r_{h_{2},1,1}$, and $r_{h_{3},1,1}$ from any subset of $d-2b = 3$ helpers $\mathcal{H}=\{h_{1}, h_{2}, h_{3}\}$, the decoder recovers $x_{f,1}$, $(x_{f,2}+e_{f}^{-1}x_{f,3})$, and $x_{f,4}$.

Similarly, in the second and third groups, the decoder uses the received repair symbols from helpers in $\mathcal{H}$ to recover $x_{f,5}$, $(x_{f,6}+e_{f}^{-1}x_{f,7})$, $x_{f,8}$, and $x_{f,9}$, $(x_{f,10}+e_{f}^{-1}x_{f,11})$, $x_{f,12}$, respectively. As a result, at the end of iteration one, entries $x_{f,1}$, $x_{f,4}$, $x_{f,5}$, $x_{f,8}$, $x_{f,9}$, and $x_{f,12}$ become "known". Moreover, for entries $x_{f,2}$, $x_{f,6}$, and  $x_{f,10}$, we label them as "inactive", since they could be recovered based on the remaining "active" entries $x_{f,3}$, $x_{f,7}$, and  $x_{f,11}$. Note  that based on this relabelling at the end of iteration one, then the remaining "active" segments are $\underline{\chi}_{f}(2)$, $\underline{\chi}_{f}(4)$, and $\underline{\chi}_{f}(6)$, which only have one "active" entry and one "known" entry each.
\end{example}

In general, as mentioned in the beginning of this subsection, in iteration $j$ of the repair procedure, any entry in the active segments is either active or known. Therefore, if $\tau_{j} < \xi$ the repair decoder has $\xi-\tau_{j}$ known entries in each of the active segments $\underline{\chi}_{f}(i),~i \in I_{g}$. Removing the known entries from the equation system in (\ref{eq_lin_eq_sys}) the repair decoder updates the system as
\begin{align}\label{eq_reduced_eq_sys}
\left[ r'_{h_{1},j,g}, \cdots , r'_{h_{(d-2b)},j,g} \right] = \left[ \underline{\chi}'_{f}(i_{1}), \cdots , \underline{\chi}'_{f}(i_{\mu_{j}}-1), \underline{\phi}'_{f,j,g} \right] \Theta'_{g},
\end{align}
where the whole vector 
\begin{align}
\left[ \underline{\chi}'_{f}(i_{1}), \cdots , \underline{\chi}'_{f}(i_{\mu_{j}}-1), \underline{\phi}'_{f,j,g} \right],  \nonumber
\end{align}
is of size $1\times (d-2b)$, and results from $[\underline{\chi}_{f}(i_{1}), \cdots , \underline{\chi}_{f}(i_{\mu_{j}}-1),\underline{\phi}_{f,j,g}]$ by removing the known entries. Also the updated coefficient matrix $\Theta'_{g}$ is derived from $\Theta_{g}$ by removing the rows corresponding to the known entries. It is easy to check that $\Theta'_{g}$ is a $(d-2b) \times (d-2b)$ invertible matrix.\footnote{The rows are linearly independent as they are rows of a Vandermonde matrix.} The repair decoder then recovers the value of all the entries in 
\begin{align}
\left[ \underline{\chi}'_{f}(i_{1}), \cdots , \underline{\chi}'_{f}(i_{\mu_{j}}-1), \underline{\phi}'_{f,j,g} \right] \nonumber
\end{align}
for each group $g$. Note that the entries in the above vector consists of all the active entries in segments $\underline{\chi}_{f}(i),~i \in I_{g} \setminus \{i_{\mu_{j}}, i_{\mu_{j}+1}\}$, along with unknown entries in $\underline{\phi}_{f,j,g}$ (category (1), and (2) as depicted in Fig. \ref{Fig_merge_segment}). As a result, all the active entries in segments $\underline{\chi}_{f}(i),~i \in I_{g} \setminus \{i_{\mu_{j}}, i_{\mu_{j}+1}\}$ will be recovered, and their labels become known. 

Let us now focus on the remaining entries in $\underline{\phi}_{f,j,g}$. Recovering each entry in category (2) reveals the value of one active entry either in $\underline{\chi}_{f}(i_{\mu_{j}})$, or in $\underline{\chi}_{f}(i_{\mu_{j}+1})$, for which then the label changes from active to known. However, entries in category (1), are formed by combining two active entries; one from $\underline{\chi}_{f}(i_{\mu_{j}})$, and the other from $\underline{\chi}_{f}(i_{\mu_{j}+1})$. Therefore, recovering the value of entries in this category, the decoder changes the labels of corresponding active entries pertaining to $\underline{\chi}_{f}(i_{\mu_{j}})$ from active to inactive, and leaves the corresponding active entries from $\underline{\chi}_{f}(i_{\mu_{j}+1})$ to remain active. It is easy to check that the number of entries in category (1) is $\tau_{j}-\sigma_{j}$. Moreover, one can easily check that the remaining active entries in segment $\underline{\chi}_{f}(i_{\mu_{j}+1})$, which are the entries participating in the formation of category (1) entries in $\underline{\phi}_{f,j,g}$, are all located at the leftmost part of $\underline{\chi}_{f}(i_{\mu_{j}+1})$. This guarantees that the invariants described in the beginning of this subsection will be preserved through the steps of the decoding.

In summary, at the end of iteration $j$ we have, 
\begin{itemize}
\item All entries in segments $\underline{\chi}_{f}(i),~i \in I_{g} \setminus \{i_{\mu_{j}}, i_{\mu_{j}+1}\}$ are recovered, for each group $g$.
\item All entries in segment $\underline{\chi}_{f}(i_{\mu_{j}})$ are either recovered or calculated in terms of a remaining active entry in $\underline{\chi}_{f}(i_{\mu_{j}+1})$.
\item The number of active entries in $\underline{\chi}_{f}(i_{\mu_{j}+1})$ is reduced from $\tau_{j}$ to $\tau_{j}-\sigma_{j}$.
\end{itemize} 

In order to start the next iteration then we simply update the value of $\tau_{j+1}$, namely the number of active entries remaining in each active segment, as
\begin{align}\label{eq_tau_update}
\tau_{j+1} = \tau_{j} - \sigma_{j}.
\end{align}

It worth mentioning that, while $\sigma_{j} > 0$, both the number of active segments as well as the number of active entries in each active segment, $\tau_{j}$, decrease in each step.

\subsubsection{Case of $\sigma_{j}=0$}

When $\sigma_{j} = 0$, then we start by taking groups of size $\mu_{j}$ active segments and do everything similar to the case of $\sigma_{j} > 0$, excepting that we do not have any merging modification on the last two segments. Therefore, (\ref{eq_repair_symbol_def}) changes to
\begin{align}\label{eq_repair_symbol_def_sigma_0}
r_{h,j,g} &= \sum_{i \in I_{j}}{\underline{\chi}_{h}(i)\underline{\varphi}_{f}^{\intercal}(i)}. 
\end{align}

Moreover, at the end of an iteration with $\sigma_{j} = 0$, all active entries in each group will be recovered, which will also result in the recovery of all the inactive entries, and the decoding ends.

\begin{example}
Let's consider the second iteration of repair for the setting described in the Example \ref{Ex_2nd_repair}. As explained in the previous examples, at the end of iteration one, the only remaining active entries are $x_{f,3}$, $x_{f,7}$, and $x_{f,11}$, and the only remaining active segments are $\underline{\chi}_{f,2}$, $\underline{\chi}_{f,4}$, and $\underline{\chi}_{f,6}$, where each of them has only one remaining active entry. This is consistent with (\ref{eq_tau_update}) as $\tau_{2} = \tau_{1}-\sigma_{1} = 2-1 = 1$. Then from (\ref{eq_mu_sigma_def}) we have $\mu_{2} = 3$, and $\sigma_{2} = 0$, and hence, we will have only one group of active segments in this iteration, with the index set $I_{1} = \{2,4,6\}$.

In this iteration, as $\sigma_{2} = 0$, we do not need any merging and, using (\ref{eq_repair_symbol_def_sigma_0}), each helper $h$ simply creates the repair symbol
\begin{align}
r_{h,2,1} &= [x_{h,3}, x_{h,4}] [e_{f}^{2}, e_{f}^{3}]^{\intercal} + [x_{h,7}, x_{h,8}] [e_{f}^{6}, e_{f}^{7}]^{\intercal} + [x_{h,11}, x_{h,12}] [e_{f}^{10}, e_{f}^{11}]^{\intercal}. \nonumber \\
&= [x_{f,3}, x_{f,4}] [e_{h}^{2}, e_{h}^{3}]^{\intercal} + [x_{f,7}, x_{f,8}] [e_{h}^{6}, e_{h}^{7}]^{\intercal} + [x_{f,11}, x_{f,12}] [e_{h}^{10}, e_{h}^{11}]^{\intercal}. \nonumber
\end{align}
However, note that $x_{f,4}$, $x_{f,8}$, and $x_{f,12}$ are known from the previous iteration and can be removed from the above equation. Then from $r_{h_{1},2,1}$, $r_{h_{2},2,1}$, and $r_{h_{3},2,1}$, provided by any subset $\mathcal{H}= \{h_{1}, h_{2}, h_{3}\}$ of helpers, the decoder forms the reduced linear equation system as described by (\ref{eq_reduced_eq_sys}) as,
\begin{align}
\left[ \begin{array}{c}
r'_{h_{1},2,1} \\
r'_{h_{3},2 ,1} \\
r'_{h_{3},2 ,1}
\end{array}\right] = \left[x_{f,3}, x_{f,7}, x_{f,11} \right] \left[ \begin{array}{c c c}
e_{h_{1}}^{2}~ & ~e_{h_{2}}^{2}~ & ~e_{h_{3}}^{2}\\
e_{h_{1}}^{6}~ & ~e_{h_{2}}^{6}~ & ~e_{h_{3}}^{6}\\
e_{h_{1}}^{10}~ & ~e_{h_{2}}^{10}~ & ~e_{h_{3}}^{10} 
\end{array}\right], \nonumber
\end{align}
and recovers the remaining active entries $x_{f,3}$, $x_{f,7}$, and $x_{f,11}$. Finally, using the equations corresponding to the inactive entries derived in the previous iteration, as explained in Example \ref{Ex_2nd_repair_step_1}, the decoder recovers the inactive entries as well and finishes the decoding.
\end{example}

The repair procedure for the presented coding scheme is summarized in the following algorithm.

\begin{algorithm}
\caption{The repair procedure}\label{Alg_repair_alternative}
\begin{algorithmic}[1]
\State Input: $d$, $f$, $b$.
\State Calculate $\xi$ using (\ref{eq_xi_def}).
\State Form all segments and initiate the label for all segments and all of their entries as "active".
\State Initiate $\tau_{1} = \xi$.
\State Set $j = 1$, and calculate $\mu_{j}$ and $\sigma_{j}$, using (\ref{eq_mu_sigma_def})
\While{$\sigma_{j}$>0}
   \State Form groups of segments each of size $\mu_{j}+1$.
   \State Calculate $m_{j}$ and $\epsilon_{j,g}$, for each group $g$, using (\ref{eq_m_ell_sigma_ell}) and (\ref{eq_epsilon_ell_j}).
   \State At each helper $h$, merge the last two segments in each group $g$, to derive $\underline{\phi}_{h,j,g}$, using (\ref{eq_phi_h_ell_j}).
   \State At each helper $h$, calculate $r_{h,j,g}$ for every group $g$ using (\ref{eq_repair_symbol_def}).
   \State At the decoder for each subset $\mathcal{H} = \{h_{1}, \cdots , h_{d-2b}\}$, form the system (\ref{eq_lin_eq_sys}), and solve.
   \State Update the labels for entries and segments.
   \State Update the value of $\tau_{j}$ using (\ref{eq_tau_update}).
   \State Update $j = j + 1$.
\EndWhile
\State Form groups of segments each of size $\mu_{j}$.
\State At each helper $h$, calculate $r_{h,j,g}$ for every group $g$ using (\ref{eq_repair_symbol_def_sigma_0}).
\State At the decoder for each subset of helpers $\mathcal{H} = \{h_{1}, \cdots , h_{d-2b}\}$, form the system (\ref{eq_lin_eq_sys}), and solve.
\State Form the estimate $\underline{\hat{x}}_{f,\mathcal{H}}$ for each subset of helpers $\mathcal{H} = \{h_{1}, \cdots , h_{d-2b}\}$.
\State Initiate all $\text{Consistency} \gets \text{False}$
\While{$\neg(\text{Consistency})$}
   \State $\mathcal{T}\gets$  A new test-group of size $d-b$
   \If{$\underline{\hat{x}}_{\mathcal{H},f} = \underline{\hat{x}}_{\mathcal{H}',f} ~ \forall{\mathcal{H}, \mathcal{H}' \subset \mathcal{T}}$}
   \State $\text{Consistency} \gets \text{True}$
   \State Output $\gets \underline{\hat{x}}_{\mathcal{H},f}$ for some $\mathcal{H}\subset \text{consistent}~\mathcal{T}$
   \EndIf
\EndWhile
\end{algorithmic}
\end{algorithm}

\subsection{Discussions}

\subsubsection{Repair bandwidth}

Now let's calculate the required repair bandwidth for the above scheme. As described in the previous subsections, in each iteration, each helper provides only one repair symbol for each group of active segments. Let $g_{j}$ denote the number of active segments groups in iteration $j$, and denote the total number of iterations by $J$. Then the total number of repair symbols provided by each helper through the repair is 
\begin{align}\label{eq_beta_sum}
\beta(d) = \sum_{j = 1}^{J}{g_{j}}.
\end{align}

Moreover, the decoder starts by setting the label "active" for all $\alpha$ entries. Then in each iteration, the decoder recovers exactly $(d-2b)$ active entries in each group of active segments, either directly or in terms of another active entry. As a result, the number of entries for which the label changes form "active" to either "known" or "inactive" in iteration $j$ is $g_{j} (d-2b)$, and we have
\begin{align}\label{eq_alpha_sum}
\alpha = \sum_{j=1}^{J}{g_{j}(d-2b)}.
\end{align}

Then from (\ref{eq_beta_sum}), and (\ref{eq_alpha_sum}) we have,
\begin{align}
\beta(d) = \frac{\alpha}{d-2b}. \nonumber
\end{align}

\subsubsection{Discussion on $\alpha$}\label{Subsec_discission_alpha}

In the provided repair scheme we require that for any $d\in D$, the number of active segments is always divisible by the number of segments in each group. In other words, denoting the total number of iterations in decoding by $J$, we require $\alpha$ to be divisible by $\xi$, $\mu_{J}$, and $(\mu_{j}+1)$ for all $j\in\{1,\cdots , J-1\}$. 

\section{Total Storage Capacity of BAER Distributed Storage Systems}
\subsection{Lower Bound}
We can now derive a lower bound on the total storage capacity based on the coding schemes presented in the previous section.

\begin{cor}\label{Cor_Achievable}
For the set of parameters $\delta$, $n$, $k$, $b$, $\alpha$, and the set $D=\{d_{1}, \cdots , d_{\delta}\}$, such that condition (\ref{eq_d_i_Cond}), and (\ref{eq_PernodeCap}) are satisfied, and the total repair bandwidth function
\begin{align}\label{eq_MBR_Achievable_RBW}
\gamma(d) = \frac{\alpha d}{d-2b},
\end{align}
the storage capacity for a BAER distributed storage system is lower bounded as
\begin{align}\label{eq_MBR_Achievable_Capasity}
F \geq \frac{\alpha (k-2 b)}{d_{\min}-2 b}\left(d_{\min}-b-\frac{(k-1)}{2}\right).
\end{align}
\end{cor}

\begin{proof}
The achievable storage capacity, $F$, of the proposed coding schemes is equal to the number of independent elements of matrix $M$. According to the structure of the matrix $M$ in (\ref{eq_Information_Matrix}), this quantity is $z$ times the total storage capacity of one MBR Product Matrix component code. Hence we have
\begin{align}
F = z\left(\frac{\kappa(\kappa +1)}{2}+\kappa(\lambda -\kappa)\right) = z \kappa \left(\lambda - \frac{(\kappa -1)}{2} \right). \nonumber
\end{align}

Replacing $\lambda = d_{\min} - 2b$, $\kappa = k - 2b$, and $z=\alpha / \delta$ we get the lower bound in (\ref{eq_MBR_Achievable_Capasity}).
\end{proof}

\subsection{Upper Bound}\label{Sec_Converse}

In the proofs of this section we will use a lemma proved in \cite{ErrRes_Salim_Journal} for the conventional regenerating codes. We restate the lemma in the setting considered in this work below while the proof follows similarly as provided in \cite{ErrRes_Salim_Journal}. 

\begin{lem}\label{Lem_Sufficiency_of_a-2b_nodes}
In a BAER regenerating code $\mathcal{C}(n$, $k$, $D$, $b$, $\alpha$, $\gamma(\cdot))$, in any data reconstruction, the data provided by any subset of size $k-2b$ of the $k$ selected nodes should be sufficient for uniquely decoding the source data stored in the network. Moreover, in any repair, the repair data provided by any subset of size $d-2b$ of the $d$ selected helpers should be sufficient for uniquely decoding the lost data.
\end{lem}

\begin{proof}
Since the proof is similar for both data reconstruction and repair processes, we will use the notation $a$ to refer to either $d$ or $k$ for the repair and data reconstruction respectively, and provide a single proof based on $a$ which works for both cases. Consider a scenario (either a repair or reconstruction) in which a set of $a$ nodes are selected to provide data. Also, assume the message (either the source data stored in the network or the content of a failed node) is $m_1\in \mathcal{M}$, where, $\mathcal{M}=\mathbb{F}_{q}^{\alpha \times \alpha}$ for data reconstruction and $\mathcal{M} = \mathbb{F}_{q}^{\alpha}$ for repair. Moreover, for any subset $\mathcal{L}$ of the $a$ selected nodes, let $\underline{y}_{\mathcal{L}}(m)$ denote the collective data provided by nodes in $L$ if the message to be recovered is $m \in \mathcal{M}$.

We provide a proof by contradiction. Assume there exists a subset of selected nodes $\mathcal{L}^{*},~|\mathcal{L}^{*}|=a-2b$ such that
\begin{align}
\underline{y}_{\mathcal{L}^{*}}(m_1)=\underline{y}_{\mathcal{L}^{*}}(m_2),~m_{1}\neq m_{2}. \nonumber
\end{align}
Note that the BAER setup allows the intruder to control any subset of nodes of size less than $b+1$. Then we can assume an intruder compromises a subset $\mathcal{L}'$ of size $b$ among the $a$ selected nodes not in $\mathcal{L}^{*}$ to provide $\underline{y}_{\mathcal{L}'}(m_2)$. Therefore the receiver will have no guarantee to recover the genuine message $m_1$, which contradicts the fact that BAER regenerating code should be capable of performing genuine repair and data reconstructions.
\end{proof}

Having this lemma along with Corollary \ref{Cor_Achievable}, we are now ready to prove Theorem \ref{THM_MBR_OptimalRBW}. 

\begin{proof}[Proof of Theorem \ref{THM_MBR_OptimalRBW}]
First note that the data stored in a single node does not have any redundancy. In other words if some part of the data stored in a single node is a function of the rest of the data we can improve the storage-bandwidth trade-off in the whole system by simply removing the redundant part from each node. Hence, $\alpha$ is an information theoretic lower bound on the required repair bandwidth. In particular, using Lemma \ref{Lem_Sufficiency_of_a-2b_nodes} we conclude that in any BAER regenerating code, the collective repair bandwidth provided by any subset of helpers of size $d-2b$ should be at least $\alpha$. As a result for any BAER regenerating code we have
\begin{align}
\beta(d)(d-2b)=\frac{\gamma(d)}{d}(d-2b) \geq \alpha,~~\forall{d \in D}. \nonumber
\end{align}

However, since Corollary \ref{Cor_Achievable} assures this is achievable by a single code for all $d \in D$, we will then have (\ref{eq_MBR_Gamma_BAER_code}) of Theorem \ref{THM_MBR_OptimalRBW}.
\end{proof}

\begin{rmrk}
Note that Theorem \ref{THM_MBR_OptimalRBW}, introduces limits for $d_{\min}$, and $k$ in an MBR BAER code. The lower limit $2b$ for $k,~d_{\min}$ could be justified using Lemma \ref{Lem_Sufficiency_of_a-2b_nodes}. Since for the choice of $d < 2b$, or $k < 2b$ there exists no subset of size $d-2b$, or $k-2b$ of nodes, and hence the storage capacity of the BAER regenerating code supporting such a $d$ or $k$ is zero (trivial code).
\end{rmrk}

\begin{lem}\label{Lem_Converse}
For any BAER regenerating code $\mathcal{C}(n$, $k$, $D$, $b$, $\alpha$, $\gamma(\cdot))$, the total storage capacity $F$ is upper bounded as follows
\begin{align}
F \leq \sum_{j=0}^{k-2 b-1}{\min\left(\alpha,\min_{d\in D}\left((d-2b-j)\frac{\gamma(d)}{d}\right)\right)}.
\end{align}
In specific, for the MBR case we have,
\begin{align}\label{eq_MBR_Capacity_UpperBound}
F_{\text{MBR}} \leq \frac{\alpha (k-2 b)}{d_{\min}-2 b}\left(d_{\min}-b-\frac{k-1}{2}\right).
\end{align}
\end{lem}

The proof is provided in Appendix \ref{App_ConverseLemma}.

Finally the proof of Theorem \ref{THM_MBR_Capacity} simply follows from (\ref{eq_MBR_Achievable_Capasity}) in Corollary \ref{Cor_Achievable}, and Lemma \ref{Lem_Converse}.

%
%



\section{Conclusion}
We considered a modified setup for the regenerating codes in which error resiliency and bandwidth adaptivity (BAER) are required to be satisfied simultaneously, and studied the storage-bandwidth trade-off in the modified BAER setup for regenerating codes. Focusing on the minimum repair bandwidth point of the storage-bandwidth tradeoff, we derived the total repair bandwidth function in the bandwidth adaptive scheme along with the corresponding storage capacity through proposing exact repair coding schemes, and providing the converse proofs. We showed that for the MBR case, optimality is achievable in strongest form (i.e., point-wise rather than Pareto optimality). We also presented an upper bound on the storage capacity of the BAER setup for the general case. 

\appendices
\section{Proof of the Correctness of Algorithm \ref{Alg_AssignmentBipartite}}\label{App_Alg_Bipartite}

\begin{lem}\label{Lem_AlgoBipartite}
The result of Algorithm \ref{Alg_AssignmentBipartite} is a bipartite graph with all the vertices in $\mathcal{V}$ having degree $d_{\min}$, and the all vertices in $\mathcal{U}$ having degree $\beta(d) = \alpha / d$, for any $d \in D$.
\end{lem}

\begin{proof}
Algorithm \ref{Alg_AssignmentBipartite} starts with all the vertices having degree zero and connects each vertex in $\mathcal{V}$ to $d_{\min}$ vertices in $\mathcal{U}$, one at a time, in iterations of the for loop. Hence, we only need to prove that all vertices in $\mathcal{U}$ will have the required degrees. Every time a vertex in $\mathcal{U}$ is connects to a new vertex in $\mathcal{V}$ its degree increases by one. Then we only need to show that the degree of all the vertices in $\mathcal{U}$ will be $\beta(d) = \alpha / d$ when the algorithm terminates. To this end, we will first show that in every iteration of the for loop, the maximum difference between the degrees of the vertices in $\mathcal{U}$ is one, using proof by contradiction.

Assume that at the end of iteration $i^{*}$ for the first time we have at least two vertices $u_{1}, u_{2} \in \mathcal{U}$, such that $\textrm{deg}(u_{2})-\textrm{deg}(u_{1})>1$. The algorithm never changes the degree of a node in $\mathcal{U}$ by more than one. Then at the end of iteration $i^{*}-1$ we should have 
\begin{align}\label{eq_deg_difference}
\textrm{deg}(u_{2})-\textrm{deg}(u_{1})=1.
\end{align}

Moreover, we conclude that in iteration $i^{*}$ the algorithm should have connected $u_{2}$ to some vertex in $\mathcal{V}$, while the degree of $u_{1}$ is remained unchanged in this iteration. However, this is contradictory since the algorithm is choosing the subset $\mathcal{W}\subset \mathcal{U}$ from the nodes with the least degree, and if $u_{1}\notin \mathcal{W}$ in iteration $i^{*}$, then $u_{1}$ could not be in $\mathcal{W}$ as from (\ref{eq_deg_difference}) we know $\textrm{deg}(u_{2})>\textrm{deg}(u_{1})$.

Next we show that in all iterations the degree of any vertex in $\mathcal{U}$ remain less than or equal to $\beta(d) = \alpha / d$. Again we use proof by contradiction. Assume in some iteration the degree of one vertex $u^{*} \in \mathcal{U}$ increases to $\beta(d)+1$. Since we have already shown that
\begin{align}
\textrm{deg}(u^{*}) - \textrm{deg}(u') \leq 1,~~\forall{u' \in \mathcal{U}} \nonumber
\end{align}
then at the end of this iteration the degree of any vertex in $\mathcal{U}$ should be at least $\beta(d)$. Therefore at the end of this iteration we have,
\begin{align}\label{eq_U_sum_degree}
\sum_{u \in \mathcal{U}}{\textrm{deg}(u)} = \sum_{u \in \mathcal{U} \setminus \{ u^{*} \}}{\textrm{deg}(u)} + \textrm{deg}(u^{*}) \geq (|\mathcal{U}|-1)\beta(d) + (\beta(d)+1) = (d-1)\frac{\alpha}{d} + \frac{\alpha}{d}+1 = \alpha + 1. 
\end{align}

However, the sum of the degrees of all the vertices in $\mathcal{U}$ should be equal to the sum of the degrees of all the nodes in $\mathcal{V}$ at the end of any iteration, and we know at the end of any iteration we have,
\begin{align}\label{eq_V_sum_degree}
\sum_{v\in\mathcal{V}}{\textrm{deg}(v)} \leq |\mathcal{V}|d_{\min} = \alpha.
\end{align}

Hence, (\ref{eq_U_sum_degree}) and (\ref{eq_V_sum_degree}) are contradictory.

Finally, since the algorithm adds $d_{\min}$ to the sum of the degrees of nodes in $\mathcal{U}$ at any iteration for $|\mathcal{V}| = \alpha / d_{\min}$ iterations, then at the end of the for loop, we have,
\begin{align}\label{eq_U_total_degree}
\sum_{u \in \mathcal{U}}{\textrm{deg}(u)} = d_{\min}|\mathcal{V}| = \alpha.
\end{align}
We have already shown that at the end of the last iteration of the loop we have,
\begin{align}\label{eq_u_total_degree}
\textrm{deg}(u)\leq \beta(d) = \frac{\alpha}{d} = \frac{\alpha}{|\mathcal{U}|},
\end{align}
Then from (\ref{eq_U_total_degree}) and (\ref{eq_u_total_degree}) we conclude that for any $u\in \mathcal{U}$, at the end of the last iteration of the loop $\textrm{deg}(u)=\beta(d)$.
\end{proof}

\section{Proof of Lemma \ref{Lem_ThetaRank}}\label{App_ThetaRank}

Recall that the code alphabet $\mathbb{F}_{q}$, $q=p^{m}$ is a simple extension field over a finite field $\mathbb{F}_{p}$, such that $p$ is a large enough prime number, and $g \in \mathbb{F}_{q}$ denotes the primitive element of $\mathbb{F}_{q}$. Also, $\mathbb{F}_{p}[x]$ denotes the ring of polynomials with coefficients from $\mathbb{F}_{p}$, and let $\varrho(x)\in \mathbb{F}_{p}[x]$ denote the minimal polynomial of $g$.
%

Consider the subset $\mathcal{H}=\{h_{1}, \cdots , h_{d-2b}\}$ of helpers, and recall that each helper node $h_{i}$ has a node specific coefficient vector $\underline{\psi}_{h_{i}}$, namely,
\begin{align}
\underline{\psi}_{h_{i}} = \left[\left(g^{h_{i}}\right)^{0}, \left(g^{h_{i}}\right)^{1}, \cdots , \left(g^{h_{i}}\right)^{\alpha -1}\right]. \nonumber
\end{align}
Without loss of generality, assume that
\begin{align}\label{eq_h_i_order}
h_{1} < h_{2} < \cdots < h_{d-2b}.
\end{align}

Also, from (\ref{eq_Omega}) the structure of the matrix $\Omega_{z_{d}}$ is given as,
\begin{align}
\Omega_{z_{d}} = \left[\begin{array}{c c c c}
\left(g^{i_{1}} \right)^{0}~ & ~\left(g^{i_{1}} \right)^{1}~ & ~\cdots~ & ~\left(g^{i_{1}} \right)^{z_{d}-1} \\
\left(g^{i_{2}} \right)^{0}~ & ~\left(g^{i_{2}} \right)^{1}~ & ~\cdots~ & ~\left(g^{i_{2}} \right)^{z_{d}-1} \\
\vdots & \vdots & \ddots & \vdots \\
\left(g^{i_{z}} \right)^{0}~ & ~\left(g^{i_{z}} \right)^{1}~ & ~\cdots~ & ~\left(g^{i_{z}} \right)^{z_{d}-1} \\
\end{array} \right], \nonumber
\end{align}
where, 
\begin{align}\label{eq_i_ell_condition_1}
i_{1} < i_{2} < \cdots < i_{z},
\end{align}
and, for any $\ell_{1}, \ell_{2}$ such that $1 \leq \ell_{1}<\ell_{2}\leq z$ we have,
\begin{align}\label{eq_i_ell_condition_2}
i_{\ell_{2}}-i_{\ell_{1}} > \alpha n.
\end{align}
We denote,
\begin{align}
\Omega_{z_{d}} = \left[\begin{array}{c}
\underline{\omega}_{1} \\
\underline{\omega}_{2} \\
\vdots \\
\underline{\omega}_{z}
\end{array}\right], \nonumber
\end{align}
where, for $j \in \{1, \cdots , z\}$, 
\begin{align}
\underline{\omega}_{j} = \left[\left(g^{i_j}\right)^{0}, \left(g^{i_j}\right)^{1}, \cdots , \left(g^{i_j}\right)^{z_{d}-1}\right].
\end{align}
Then from the definitions of the matrices $\Phi_{h},~h \in \mathcal{H}$, given by (\ref{eq_PhiMatrix}), we have,
\begin{align}
\Phi_{h}\Omega_{z_{d}} &= \left[\begin{array}{c}
\underline{\omega}_{1}\otimes \underline{\psi}_{h}^{\intercal}(1) \\
\underline{\omega}_{2}\otimes \underline{\psi}_{h}^{\intercal}(2) \\
\vdots \\
\underline{\omega}_{z}\otimes \underline{\psi}_{h}^{\intercal}(z) 
\end{array}\right] 
=\left[\begin{array}{c c c}
\left(g^{i_1}\right)^{0} \underline{\psi}_{h}^{\intercal}(1) & ~\cdots ~ & \left(g^{i_1}\right)^{z_{d}-1} \underline{\psi}_{h}^{\intercal}(1) \\
\left(g^{i_2}\right)^{0} \underline{\psi}_{h}^{\intercal}(2) & ~\cdots ~ & \left(g^{i_2}\right)^{z_{d}-1} \underline{\psi}_{h}^{\intercal}(2) \\
\vdots & \ddots & \vdots \\
\left(g^{i_z}\right)^{0} \underline{\psi}_{h}^{\intercal}(z) & ~\cdots ~ & \left(g^{i_z}\right)^{z_{d}-1} \underline{\psi}_{h}^{\intercal}(z)
\end{array}\right]. \nonumber
\end{align}
Therefore, by stacking matrices $\Phi_{h_{i}}\Omega_{z_{d}}$ for $h_{i}\in\mathcal{H}$, we obtain,
\begin{align}\label{eq_Theta_in_Xi}
\Theta_{\mathcal{H}} &= \left[\Phi_{h_{1}} \Omega_{z_{d}}^{\intercal},\cdots,\Phi_{h_{d-2b}} \Omega_{z_{d}}^{\intercal} \right] 
= \Pi \left[\begin{array}{c}
\underline{\omega}_{1} \otimes V_{1} \\
\underline{\omega}_{2} \otimes V_{2} \\
\vdots \\
\underline{\omega}_{z} \otimes V_{z} 
\end{array}\right]  
= \Pi \left[\begin{array}{c c c}
\left(g^{i_1}\right)^{0} V_{1} & ~ \cdots ~ & \left(g^{i_1}\right)^{z_{d}-1} V_{1} \\
\left(g^{i_2}\right)^{0} V_{2} & ~ \cdots ~ & \left(g^{i_2}\right)^{z_{d}-1} V_{2} \\
\vdots & \ddots & \vdots \\
\left(g^{i_z}\right)^{0} V_{z} & ~ \cdots ~ & \left(g^{i_z}\right)^{z_{d}-1} V_{z}
\end{array}\right], 
\end{align}
where $\Pi$ is an appropriate column permutation matrix and the matrices $V_{\ell},~\ell \in\{1, \cdots , z \}$ are transposed Vandermonde matrices of size $(d_{\min}-2b) \times (d-2b)$ as,
\begin{align}
V_{\ell} = \left[\begin{array}{c c c c}
\underline{\psi}_{h_{1}}^{\intercal}(\ell) ~ & ~ \underline{\psi}_{h_{2}}^{\intercal}(\ell) ~ & \cdots & ~ \underline{\psi}_{h_{d-2b}}^{\intercal}(\ell)
\end{array}\right]. \nonumber
\end{align}

Note that multiplying $\Pi$ from right does not change the rank. Hence, in order to show that $\Theta_{\mathcal{H}}$ is invertible it suffices to prove the determinant of the following matrix is non-zero,
\begin{align}\label{eq_Xi_def}
\Xi_{\mathcal{H}} = \hspace{-1mm} \left[\begin{array}{c}
\underline{\omega}_{1} \otimes V_{1} \\
\underline{\omega}_{2} \otimes V_{2} \\
\vdots \\
\underline{\omega}_{z} \otimes V_{z} 
\end{array}\right]\hspace{-1mm} = \hspace{-1mm}\left[\begin{array}{c c c}
\left(g^{i_1}\right)^{0} V_{1} & ~ \cdots ~ & \left(g^{i_1}\right)^{z_{d}-1} V_{1} \\
\left(g^{i_2}\right)^{0} V_{2} & ~ \cdots ~ & \left(g^{i_2}\right)^{z_{d}-1} V_{2} \\
\vdots & \ddots & \vdots \\
\left(g^{i_z}\right)^{0} V_{z} & ~ \cdots ~ & \left(g^{i_z}\right)^{z_{d}-1} V_{z}
\end{array}\right]\hspace{-1mm}. 
\end{align}
The sketch of the proof is as follows:
\begin{itemize}
\item We show that every element in $\Xi_{\mathcal{H}}$ can be represented in the form of an exponent of $g$.
\item We show the determinant of $\Xi_{\mathcal{H}}$ can be represented as a polynomial $f(x)\in \mathbb{F}_{p}[x]$, evaluated at $g$.
\item We show $f(x)$ is non-trivial.
\item We show that (for large enough field size) $g$ can not be a root of $f(x)$ and hence the determinant is non-zero.
\end{itemize}

Some of the ideas used in this proof are similar to \cite{MDP_Almeida}. 

Observe that each block $\left(g^{i_\ell}\right)^{j} V_{\ell}$ in $\Xi_{\mathcal{H}}$ is a $(d_{\min}-2b) \times (d-2b)$ submatrix. Then every entry of the matrix $\Xi_{\mathcal{H}}$, is a product of two powers of the primitive element $g$. In particular for the element $\Xi_{\mathcal{H}}(r,c)$, located in row $r$ and column $c$, we have
\begin{align}\label{eq_Xi_ij}
\Xi_{\mathcal{H}}(r,c) &= \left( g^{i_{\ell}} \right)^{\ell '} V_{\ell}(r-(d_{\min}-2b)\ell , \ell' +1) \nonumber \\
&= \left( g^{i_{\ell}} \right)^{\ell '} \left( g^{h_{\ell ''}} \right)^{(r-1)} \nonumber \\
&= g^{\ell ' i_{\ell} + (r-1)h_{\ell ''}},
\end{align}
where,
\begin{align}
\ell = \lceil \frac{r}{d_{\min}-2b}\rceil , ~~ \ell ' = \lceil \frac{c}{d-2b} \rceil -1 , ~~ \ell '' = c \mod (d-2b). \nonumber
\end{align}


Denoting the set of all permutations on $\{1, \cdots , \alpha \}$ by $\mathcal{S}_{\alpha}$, the Leibniz extension for the determinant is given by,
\begin{align}\label{eq_Leibniz_first}
|\Xi_{\mathcal{H}}| = \sum_{\sigma \in \mathcal{S}_{\alpha}}\left((-1)^{\textrm{sgn}(\sigma)} \prod_{i=1}^{\alpha}{\Xi_{\mathcal{H}}(i,\sigma(i))}\right).
\end{align}
From (\ref{eq_Xi_ij}) we know every element $\Xi_{\mathcal{H}}(i,\sigma(i))$ can be represented as an exponent of $g$. Then for any $\sigma \in \mathcal{S}_{\alpha}$, we use $\textrm{pow}(\sigma)$ to denote the exponent of $g$ associated with the permutation $\sigma$ as follows
\begin{align}
g^{\textrm{pow}(\sigma)} = \prod_{i=1}^{\alpha}{\Xi_{\mathcal{H}}(i,\sigma(i))}, \nonumber
\end{align}
and we can rewrite (\ref{eq_Leibniz_first}) as
\begin{align}\label{eq_Leibniz}
|\Xi_{\mathcal{H}}| = \sum_{\sigma \in \mathcal{S}_{\alpha}} \left((-1)^{\textrm{sgn}(\sigma)} g^{\textrm{pow}(\sigma)}\right). 
\end{align}

Therefore, from (\ref{eq_Theta_in_Xi}) and (\ref{eq_Xi_ij}), for large enough prime\footnote{p should be large enough to make sure any coefficient in the polynomial is smaller than $p$, and hence is an element of $\mathbb{F}_{p}$. However, the number of terms in the determinant expansion is finite for any finite $\alpha$, and hence is the largest possible coefficient.}  $p$, it is clear that there exists a polynomial $f(x) \in \mathbb{F}_{p}[x]$, such that
\begin{align}\label{eq_f_x}
|\Xi_{\mathcal{H}} | = f(g). 
\end{align}
We refer to the polynomial $f(g)$ as the \emph{determinant polynomial} in the rest of this appendix.

Note that in (\ref{eq_Leibniz}), each term is associated with one of the permutations on the set $\{1,\cdots , \alpha\}$. In the rest of this appendix we consider each permutation on the set $\{1,\cdots , \alpha\}$ as a bijective mapping from the set of rows of the matrix $\Xi_{\mathcal{H}}$ to the set of its columns. Moreover, we refer to the $i^{\text{th}}$ group of $d_{\min}-2b$ consequent rows in matrix $\Xi_{\mathcal{H}}$ as the $i^{\text{th}}$ \emph{block row}. Similarly, the $i^{\text{th}}$ group of $d_{\min}-2b$ consequent columns in matrix $\Xi_{\mathcal{H}}$ is referred to as the $i^{\text{th}}$ \emph{block column}, as depicted in Fig. \ref{Fig_block_rows_cols}.

\begin{figure}
\centering
\resizebox{4 in}{!}{\includegraphics{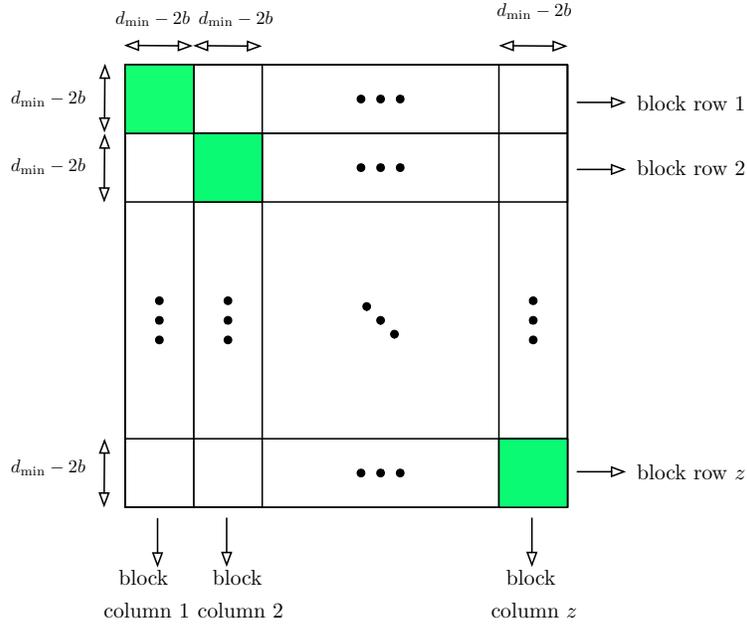}}
\caption{The partitioning of rows and columns of the $\Xi_{\mathcal{H}}$ matrix into block rows and block columns, along with the green diagonal submatrices containing all the entries $\Xi_{\mathcal{H}}(j,\sigma(j))$, for any $\sigma \in \mathcal{F}$.}\label{Fig_block_rows_cols}
\end{figure}

Finally, let $\mathcal{F}$ denote the family of permutations on the set $\{1,\cdots , \alpha\}$, which are mapping the rows in each block row $i$ to the columns in the block column $i$, for each $i\in \{1,\cdots , z_{d}\}$. In other words for any permutation $\sigma \in \mathcal{F}$, all entries $\Xi_{\mathcal{H}}(j,\sigma(j))$ are located in the green squares in Fig. \ref{Fig_block_rows_cols}.

\begin{lem}\label{Lem_permutation_family}
For any permutation $\sigma$ on the set $\{1, \cdots , \alpha\}$, if $\sigma \notin \mathcal{F}$, then there exists another permutation $\sigma'$, such that,
\begin{align}\label{eq_pov_comp}
\textrm{pow}(\sigma) < \textrm{pow}(\sigma'). 
\end{align}
\end{lem}

\begin{proof}
We prove this lemma by constructing the permutation $\sigma'$ based on the permutation $\sigma \notin \mathcal{F}$, such that $\sigma'$ is different from $\sigma$ in exactly two pairs of rows and columns, and (\ref{eq_pov_comp}) is satisfied.

Let $i$ denote the first block row such that $\sigma$ maps some row $r_{1}$ in block row $i$ to some column $c_{1}$ in a block column $j \neq i$. Then from the definition of the family $\mathcal{F}$ we conclude $j>i$. Moreover, since the mapping induced by any permutation is bijective and the size of block rows and block columns are the same, then there should exist a column $c_{2}$ in the block column $i$ which is mapped by the permutation $\sigma$ to some row $r_{2}$ in a row block $i' \neq i$. Again, we conclude that $i' > i$, and hence we conclude
\begin{align}\label{eq_r_1_r_2}
r_{2} > r_{1}.
\end{align}

To summarize, we have
\begin{align}
\sigma(r_{1}) = c_{1},~\sigma(r_{2}) = c_{2}. \nonumber
\end{align}

We now claim that the permutation $\sigma'$ defined as follows satisfies (\ref{eq_pov_comp}),
\begin{align}\label{eq_sigma_ell}
\sigma'(r) = \left\lbrace \begin{array}{c c}
c_{2}~ & ~r = r_{1}, \\
c_{1}~ & ~r = r_{2}, \\
\sigma(r)~ & ~\text{otherwise}.
\end{array} \right.
\end{align}

Figure \ref{Fig_switch} depicts the process of deriving $\sigma'$ from $\sigma$.

\begin{figure}
\centering
\resizebox{4 in}{!}{\includegraphics{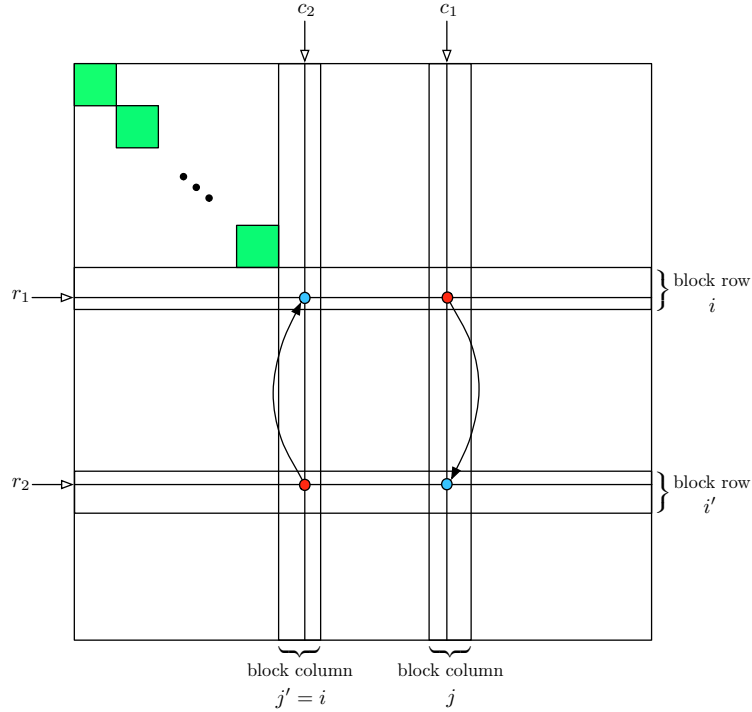}}
\caption{The process of deriving the permutation $\sigma'$ from $\sigma$, by switching between the the pair of entries denoted by red circles and the pair of entries denoted by blue circles.}\label{Fig_switch}
\end{figure}

Note that the exponent of the term associated with a permutation $\sigma$ in $f(x)$  is the sum of all the exponents of $g$ in the entries $\Xi_{\mathcal{H}}(r,\sigma(r))$. From (\ref{eq_sigma_ell}) it is clear that entries $\Xi_{\mathcal{H}}(r,\sigma(r))=\Xi_{\mathcal{H}}(r,\sigma'(r))$, for all $r$ except $r_{1}$, and $r_{2}$. Moreover, from (\ref{eq_Xi_ij}), for some integers $\ell_{1}, \ell_{2} \in \{1,\cdots , d-2b\}$, we have
\begin{align}
\Xi_{\mathcal{H}}(r_{1},\sigma(r_{1})) = \left(g^{i_{i}} \right)^{\lfloor \frac{c_{1}}{d-2b}\rfloor} \left(g^{h_{\ell_{1}}} \right)^{(r_{1}-1)}, \nonumber \\
\Xi_{\mathcal{H}}(r_{2},\sigma(r_{2})) = \left(g^{i_{i'}} \right)^{\lfloor \frac{c_{2}}{d-2b}\rfloor} \left(g^{h_{\ell_{2}}} \right)^{(r_{2}-1)}. \nonumber \\
\end{align}
Similarly from (\ref{eq_sigma_ell}) we have
\begin{align}
\Xi_{\mathcal{H}}(r_{1},\sigma'(r_{1})) = \left(g^{i_{i}} \right)^{\lfloor \frac{c_{2}}{d-2b}\rfloor} \left(g^{h_{\ell_{2}}} \right)^{(r_{1}-1)}, \nonumber \\
\Xi_{\mathcal{H}}(r_{2},\sigma'(r_{2})) = \left(g^{i_{i'}} \right)^{\lfloor \frac{c_{1}}{d-2b}\rfloor} \left(g^{h_{\ell_{1}}} \right)^{(r_{2}-1)}. \nonumber \\
\end{align}
Note that since $j > i$, then we conclude that $c_{1} > c_{2}$, and as a result,
\begin{align}
\lfloor \frac{c_{1}}{d-2b}\rfloor \geq \lfloor \frac{c_{2}}{d-2b}\rfloor. \nonumber
\end{align}

Then we have
\begin{align}
\textrm{pow}(\sigma')-\textrm{pow}(\sigma) = (i_{i'}-i_{i})\left(\lfloor \frac{c_{1}}{d-2b}\rfloor - \lfloor \frac{c_{2}}{d-2b}\rfloor \right) + (h_{\ell_{1}}-h_{\ell_{2}})(r_{2}-r_{1}).
\end{align}
Now, we have two cases as follows. 

Case 1:
\begin{align}
\lfloor \frac{c_{1}}{d-2b}\rfloor - \lfloor \frac{c_{2}}{d-2b}\rfloor > 0. \nonumber
\end{align}
In this case, note that $h_{\ell_{1}}, h_{\ell_{2}} \in \{1, \cdots , n\}$, and $r_{1}, r_{2} \in \{1, \cdots , \alpha\}$, then we have,
\begin{align}
(h_{\ell_{1}}-h_{\ell_{2}})(r_{2}-r_{1}) \geq - \alpha n, \nonumber
\end{align}
and hence using (\ref{eq_i_ell_condition_1}), and (\ref{eq_i_ell_condition_2}) we have $\textrm{pow}(\sigma')-\textrm{pow}(\sigma) > 0$.

Case 2:
\begin{align}
\lfloor \frac{c_{1}}{d-2b}\rfloor - \lfloor \frac{c_{2}}{d-2b}\rfloor = 0. \nonumber
\end{align}
In this case, considering the structure of the matrix $\Xi_{\mathcal{H}}$, as presented in (\ref{eq_Xi_def}), from (\ref{eq_h_i_order}) we conclude that $h_{\ell_{2}} > h_{\ell_{1}}$. Therefore, from (\ref{eq_r_1_r_2}) we have $(h_{\ell_{1}}-h_{\ell_{2}})(r_{2}-r_{1}) > 0$, and hence $\textrm{pow}(\sigma')-\textrm{pow}(\sigma) > 0$.
\end{proof}

\begin{cor}\label{Cor_perm_family}
In the representation of the determinant of matrix $\Xi_{\mathcal{H}}$ introduced in (\ref{eq_Leibniz}), the term(s) with highest exponent of $g$ are associated to permutation(s) in the family $\mathcal{F}$.
\end{cor}

\begin{lem}\label{Lem_Unique_Largest_exponent}
For the matrix $\Xi_{\mathcal{H}}$, as introduced in (\ref{eq_Xi_def}), let $f(x)$ denote the polynomial introduced in (\ref{eq_f_x}). Then the term with highest exponent of $g$ in $f(g)$ is unique. 
\end{lem}

\begin{proof}
Based on the Corollary \ref{Cor_perm_family}, the highest exponent of $g$ appears in permutation(s) from the family $\mathcal{F}$. Hence, to prove this lemma we only need to show there exists a unique permutation $\sigma^{*}\in \mathcal{F}$, such that $\textrm{pow}(\sigma^{*})$ is the maximum among all other permutations in $\mathcal{F}$. To this end, we first note that from (\ref{eq_Xi_ij}) it is clear that every entity of the matrix $\Xi_{\mathcal{H}}$ can be considered as the product of two components. The first component is in the form 
\begin{align}\label{eq_Xi_ij_omega_component}
\left(g^{i_{\ell}}\right)^{\ell'},
\end{align}
for some $\ell\in \{1, \cdots , z\}$, and $\ell' \in \{0, \cdots , z_{d}-1\}$. We refer to this component as the \emph{omega} component. The other component is in the form
\begin{align}\label{eq_Xi_ij_phi_component}
\left(g^{h_{\ell''}} \right)^{r-1},
\end{align}
for some $h_{\ell''} \in \mathcal{H}$, and $r\in \{1, \cdots , \alpha\}$. We refer to this component as the \emph{phi} component. 

Likewise, each term in $f(g)$, associated with a permutation $\sigma$, which is written as
\begin{align}
\prod_{i=1}^{\alpha}{\Xi_{\mathcal{H}}(i,\sigma(i))}, \nonumber
\end{align}
can be decomposed into two components. One component is the product of all the omega components of the entities $\Xi_{\mathcal{H}}(i,\sigma(i))$, and the we refer to it as the \emph{omega} component, and the other part is the product of all the phi components of $\Xi_{\mathcal{H}}(i,\sigma(i))$'s and is referred to as the \emph{phi} component. 

Since any permutation in $\mathcal{F}$ maps all the rows in block row $i$ to all the columns in block column $i$, it is then easy to check that the omega component of all the terms associated with permutations in $\mathcal{F}$ are the same. To complete the proof we then only need to show that the phi component of the term associated with a unique permutation $\sigma^{*}\in \mathcal{F}$ has the largest exponent of $g$. 

From the structure of the matrix $\Xi_{\mathcal{H}}$, as presented in (\ref{eq_Xi_def}), we can see that each column is associated with one of the helpers $h_{\ell}\in \mathcal{H}$, such that the exponent of $g$ in the phi component of all the entities in that column is a multiple of $h_{\ell}$. Let's refer to such $h_{\ell}$ as the \emph{helper index} of the column. Moreover, from (\ref{eq_Xi_ij_phi_component}) and using the "rearrangement inequality" \cite{Rearrangement} Section 10.2, Theorem 368, we conclude that the exponent of $g$ in the phi component of a term associated with a permutation $\sigma \in \mathcal{F}$ is maximized when in each block row $i$, the helper index of the column assigned to each row $r$ increases with as $r$ increases. Note that, this mapping in each block row is well defined since the size of each block column is $d_{\min} - 2b \leq d - 2b$ for any $d \in D$, and hence each helper index only appears at most once in each block column. Therefore, the permutation defined based on the maximizing assignment is unique in $\mathcal{F}$. This completes the proof.
\end{proof}

The following example illustrates the unique permutation associated with the maximum exponent of $g$ in the expression of the determinant $|\Xi_{\mathcal{H}}|$ as represented in (\ref{eq_Leibniz}).

\begin{example}
Consider $n=6$, $k=4$, $D=\{4,5\}$, and $b = 1$. Then notice that $d_{\min} = 4$, and setting $\alpha = 6$ satisfies (\ref{eq_PernodeCap}. Assume node $f=1$ is failed and consider a repair based on $d = 5$ helpers. We then focus on a subset $\mathcal{H} = \{h_{1},h_{2},h_{3}\}$ of helpers of size $d-2b = 3$ for this example, such that $1\leq h_{1}<h_{2}<h_{3} \leq n$. From (\ref{eq_Xi_def}) we have,
\begin{align}
\Xi_{\mathcal{H}} = \left[ \begin{array}{c c}
V_{1}~ & ~g^{i_{1}}V_{1} \\
V_{2}~ & ~g^{i_{2}}V_{2} \\
V_{3}~ & ~g^{i_{3}}V_{3}
\end{array}
\right], \nonumber
\end{align}
for some $i_{1}$, $i_{2}$, and $i_{3}$ satisfying (\ref{eq_i_ell_condition_1}), and (\ref{eq_i_ell_condition_2}), and
\begin{align}
V_{1} = \left[ \begin{array}{c c c}
\left(g^{h_{1}} \right)^0~ & ~\left(g^{h_{2}} \right)^0~ & ~\left(g^{h_{3}} \right)^0 \\
\left(g^{h_{1}} \right)^1~ & ~\left(g^{h_{2}} \right)^1~ & ~\left(g^{h_{3}} \right)^1
\end{array}\right],~~V_{2} = \left[ \begin{array}{c c c}
\left(g^{h_{1}} \right)^2~ & ~\left(g^{h_{2}} \right)^2~ & ~\left(g^{h_{3}} \right)^2 \\
\left(g^{h_{1}} \right)^3~ & ~\left(g^{h_{2}} \right)^3~ & ~\left(g^{h_{3}} \right)^3
\end{array}\right],~~V_{3} = \left[ \begin{array}{c c c}
\left(g^{h_{1}} \right)^4~ & ~\left(g^{h_{2}} \right)^4~ & ~\left(g^{h_{3}} \right)^4 \\
\left(g^{h_{1}} \right)^5~ & ~\left(g^{h_{2}} \right)^5~ & ~\left(g^{h_{3}} \right)^5
\end{array}\right].   \nonumber
\end{align}
Note that in this example $z_{d} = \alpha/ (d-2b) = 6/3 = 2$, and each submatrix $V_{j},~j\in\{1,2,3\}$ is of size $(d_{\min}-2b)\times(d-2b) = 2 \times 3$. Moreover, each block row and block column is of size $d_{\min}-2b = 2$. Finally the elements $\Xi_{\mathcal{H}}(i,\sigma^{*}(i))$, for the permutation $\sigma^{*}$ corresponding to the maximum exponent of $g$ in (\ref{eq_Leibniz}) are illustrated in red as follows

\begin{align}
\Xi_{\mathcal{H}} = \left[\begin{array}{c c} 
\left[ \begin{array}{c c | c}
{\color{red}\left(g^{h_{1}} \right)^0}~ & ~\left(g^{h_{2}} \right)^0~ & ~\left(g^{h_{3}} \right)^0 \\
\left(g^{h_{1}} \right)^1~ & ~{\color{red}\left(g^{h_{2}} \right)^1}~ & ~\left(g^{h_{3}} \right)^1
\end{array}\right]&\left[ \begin{array}{c | c c}
\left(g^{h_{1}} \right)^0 g^{i_{1}}~ & ~\left(g^{h_{2}} \right)^0 g^{i_{1}}~ & ~\left(g^{h_{3}} \right)^0 g^{i_{1}} \\
\left(g^{h_{1}} \right)^1 g^{i_{1}}~ & ~\left(g^{h_{2}} \right)^1 g^{i_{1}}~ & ~\left(g^{h_{3}} \right)^1 g^{i_{1}}
\end{array}\right] \\ \hline 
\left[ \begin{array}{c c | c}
\left(g^{h_{1}} \right)^2~ & ~\left(g^{h_{2}} \right)^2~ & ~\left(g^{h_{3}} \right)^2 \\
\left(g^{h_{1}} \right)^3~ & ~\left(g^{h_{2}} \right)^3~ & ~{\color{red}\left(g^{h_{3}} \right)^3}
\end{array}\right]
&\left[ \begin{array}{c | c c}
{\color{red}\left(g^{h_{1}} \right)^2 g^{i_{2}}}~ & ~\left(g^{h_{2}} \right)^2 g^{i_{2}}~ & ~\left(g^{h_{3}} \right)^2 g^{i_{2}} \\
\left(g^{h_{1}} \right)^3 g^{i_{2}}~ & ~\left(g^{h_{2}} \right)^3 g^{i_{2}}~ & ~\left(g^{h_{3}} \right)^3 g^{i_{2}}
\end{array}\right] \\ \hline
\left[ \begin{array}{c c | c}
\left(g^{h_{1}} \right)^4~ & ~\left(g^{h_{2}} \right)^4~ & ~\left(g^{h_{3}} \right)^4 \\
\left(g^{h_{1}} \right)^5~ & ~\left(g^{h_{2}} \right)^5~ & ~\left(g^{h_{3}} \right)^5
\end{array}\right]
&\left[ \begin{array}{c | c c}
\left(g^{h_{1}} \right)^4 g^{i_{2}}~ & ~{\color{red}\left(g^{h_{2}} \right)^4 g^{i_{2}}}~ & ~\left(g^{h_{3}} \right)^4 g^{i_{2}} \\
\left(g^{h_{1}} \right)^5 g^{i_{2}}~ & ~\left(g^{h_{2}} \right)^5 g^{i_{2}}~ & ~{\color{red}\left(g^{h_{3}} \right)^5 g^{i_{2}}}
\end{array}\right]
\end{array}
\right] \nonumber
\end{align}

As illustrated above, in the second row and column blocks, the permutation $\sigma^{*}$ assigns the forth row of $\Xi_{\mathcal{H}}$  to column three, and the third row to column four, since column three has helper index $h_{3}$ which is by assumption larger than the helper index of column four, namely $h_{1}$.
\end{example}

Note that in the Leibniz expansion for the determinant of $\Xi_{\mathcal{H}}$, each term is a power of the primitive element of $\mathbb{F}_{q}$, namely $g$. Moreover, using the result of Lemma \ref{Lem_Unique_Largest_exponent} we conclude that the term associated with $\sigma(i) = i$, provides the largest exponent of $g$ which is unique, and hence does not get cancelled by any other term in the expansion. Then the polynomial $f(x)$ is non-trivial as the coefficient of the term with highest power in $f(x)$ is one. Let $\textrm{deg}(f)$ denote the degree of the polynomial $f(x)$, and assume the degree of the expansion field $\mathbb{F}_{q}$ over $\mathbb{F}_{p}$, namely $m$, is larger than $\textrm{deg}(f)$. In other words, assume
\begin{align}\label{eq_field_size}
\log_{p}(q) > \textrm{deg}(f).
\end{align}
 
Now we use the following well-known result to show that the $g$ could not be a root of the polynomial $f(x)$.

\begin{thm}\label{Thm_divisible}\cite{Algebra_book}
Let $\mathbb{F}_{q}$ be a finite field with $q=p^{m}$ elements, for some prime number $p$. Also let $g$ be a primitive element in $\mathbb{F}_{q}$, with the minimal polynomial $\varrho(x)\in \mathbb{F}_{p}[x]$. If $f(x)\in\mathbb{F}_{p}[x]$ with $f(g) = 0$, then $\varrho(x)|f(x)$.
\end{thm}

The above result holds due to the fact that the minimal polynomial is irreducible, and the ring of polynomials is an integral domain.

Using (\ref{eq_field_size}), it is clear that the polynomial $f(x)$ is not divisible by the minimal polynomial $\varrho(x)$, and hence using Theorem \ref{Thm_divisible}, we conclude that the determinant of the matrix $\Xi_{\mathcal{H}}$ is non-zero, which in turn proves the non-singularity of $\Theta_{\mathcal{H}}$.

\begin{rmrk}[Field size requirement]
Note that in order to satisfy (\ref{eq_field_size}) and guarantee $f(x) \in \mathbb{F}_{p}[x]$, we need both $p$ and $q$ to be large enough. For the sake of completeness here we provide lower bounds on each one that guarantee the required conditions, although they might not necessarily be the tightest lower bounds.

Having 
\begin{align}
p > \alpha ! \nonumber
\end{align}
guarantees $f(x) \in \mathbb{F}_{p}[x]$. Moreover, selecting $i_{1}, \cdots , i_{z}$ as 
\begin{align}
i_{j} = \alpha n (j-1) + 1, \nonumber
\end{align}
satisfies both (\ref{eq_i_ell_condition_1}) and (\ref{eq_i_ell_condition_2}). The largest exponent of $g$, realized in the term associated with the permutation $\sigma^{*}$, is then upper bounded by
\begin{align}
\alpha^{2} n \left(1+z^2 \right). \nonumber
\end{align}
Finally the upper bound on the required field size is 
\begin{align}
\left( \alpha ! \right)^{\left(\alpha ^2 n\left(1+z^2 \right) \right)}. \nonumber
\end{align}

This is of course very huge for practical settings, however, note that the goal of this appendix is to provide a proof for the the non-singularity of the matrix $\Theta_{\mathcal{H}}$. We present a different coding scheme for practical settings in Section \ref{Sec_Coding_Sch_II}, which reduces the field size requirement to $n$.
\end{rmrk}


\section{Proof of Lemma \ref{Lem_Converse}}\label{App_ConverseLemma}

\begin{proof}
The proof follows ideas similar to \cite{Regenerating}, \cite{ErrRes_Salim_Journal}, and \cite{BWAdapt_Opportunistic}. However, to derive an upper bound on the capacity of a BAER setting, we introduce a \emph{genie-aided} version of this code. Then we derive the upper bound on the capacity $F$, by finding an appropriate cut-set in the \emph{information flow graph} corresponding to the genie-aided version.

In the genie-aided version of $\mathcal{C}(n$, $k$, $D$, $b$, $\alpha$, $\gamma(\cdot))$, when we select the set of $d$ helper nodes for a repair, the genie identifies a subset of size $d-2b$ of the selected helpers as \emph{genuine} helpers, and we will only receive repair data from them. Similarly, in the download process after choosing the set of $k$ nodes, the genie identifies a subset of size $k-2b$ of genuine nodes among them, and the data collector only collects data from this subset. From Lemma \ref{Lem_Sufficiency_of_a-2b_nodes} we know that limiting the connections in the genie-aided version will not reduce the storage capacity. Hence, the storage capacity of the genie-aided version is an upper bound for $F$; the storage capacity of the original setting.

\begin{figure*}[!ht]
\psfragscanon
\psfrag{A}[cc][cc][1.5][0]{$\text{S}$}
\psfrag{B}[cc][cc][1.5][0]{{$\text{DC}$}}
\psfrag{C}[cc][cc][1.5][0]{{$1_{\text{in}}$}}
\psfrag{D}[cc][cc][1.5][0]{{$1_{\text{out}}$}}
\psfrag{E}[cc][cc][1.5][0]{{$2_{\text{in}}$}}
\psfrag{F}[cc][cc][1.5][0]{{$2_{\text{out}}$}}
\psfrag{G}[cc][cc][1.5][0]{{$n_{\text{in}}$}}
\psfrag{H}[cc][cc][1.5][0]{{$n_{\text{out}}$}}
\psfrag{I}[cc][cc][1.5][0]{{${\ell_1}_{\text{in}}$}}
\psfrag{J}[cc][cc][1.5][0]{{${\ell_1}_{\text{out}}$}}
\psfrag{K}[cc][cc][1.5][0]{{${\ell_2}_{\text{in}}$}}
\psfrag{L}[cc][cc][1.5][0]{{${\ell_2}_{\text{out}}$}}
\psfrag{S}[cc][cc][1.5][0]{{${\ell_{k-2b}}_{\text{in}}$}}
\psfrag{M}[cc][cc][1.5][0]{{${\ell_{k-2b}}_{\text{out}}$}}
\psfrag{O}[cc][cc][1.5][0]{{$\infty$}}
\psfrag{P}[cc][cc][1.5][0]{{$d-2b~\Bigg\{$}}
\psfrag{Q}[cc][cc][1.5][0]{{$d'-2b-1~\Bigg\{$}}
\psfrag{R}[cc][cc][1.5][0]{{$d''-k+1~\Bigg\{$}}
\psfrag{N}[cc][cc][1.5][0]{{$\alpha$}}
\centering
\resizebox{5 in}{!}{
\includegraphics[keepaspectratio=true,scale=1,trim={0 10cm 0 0},clip]{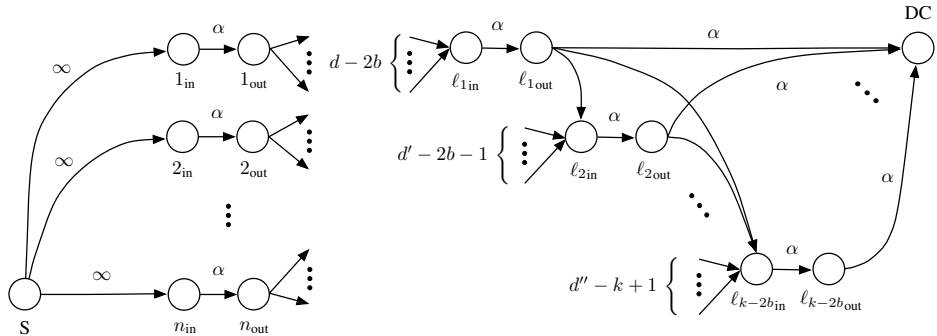}}
\caption{The information flow graph of the genie-aided version of a BAER regenerating code $\mathcal{C}(n$, $k$, $D$, $b$, $\alpha$, $\gamma(\cdot))$.}\label{Fig_InformationFlowGraph}
\end{figure*}

To derive an upper bound on the storage capacity of the genie-aided version, we will consider the \emph{information flow graph} as introduced in \cite{Regenerating}. The  information flow graph is a \emph{directed acyclic graph} (DAG) model to represent the flow of information during a sequence of repairs and data reconstructions in the network. The source of information is represented as a single node which has only out-going edges, and any data collector is represented as a single node which only has in-coming edges. Every storage node $\ell$, which has once been used in the network, is represented by a pair of nodes $\ell_{\text{in}}$, and $\ell_{\text{out}}$ in the DAG such that an edge with capacity $\alpha$ takes the flow of information from $\ell_{\text{in}}$ to $\ell_{\text{out}}$. This edge represents the per node storage capacity constraint for node $\ell$, hence we refer to such edges as \emph{storage edges}. In addition to storage edges there are three other types of edges in the information flow graph, namely the \emph{download edges}, the \emph{repair edges}, and the \emph{source edges}. Download edges have capacity $\alpha$ and take information flow from $\ell_{\text{out}}$ to a data collector node if $\ell$ is among the genuine selected nodes for the data collector. Repair edges take information flow from $\ell_{\text{out}}$ to $\ell'_{\text{in}}$ if $\ell'$ is a replacement node in the distributed storage network, and $\ell$ is one of the selected genuine helpers for the repair. The capacity of repair edges in the information flow graph is then $\gamma(d)/(d-2b)$, for the chosen parameter $d$ in the corresponding repair. Finally, we also consider a set of $n$ source edges with infinite capacity, taking information flow from the source node to the input node of initial $n$ storage nodes in the network. Figure \ref{Fig_InformationFlowGraph} depicts one example of an information flow graph.

Corresponding to any specific sequence of repair and data reconstruction processes, there exists a specific information flow graph. Any cut-set in the DAG model for an information flow graph consists of a set of edges such that after removing them, there is no path from the source to the data collector. As a result, the sum capacity of all the edges in a cut-set provides an upper bound on the capacity of information which could be stored in the corresponding distributed storage network and restored by a data collector after the sequence of repairs associated to the information flow graph is performed. We are going to consider the information flow graph depicted in Fig. \ref{Fig_InformationFlowGraph}.

As depicted in the Fig \ref{Fig_InformationFlowGraph} in our scenario a data collector is downloading the data stored in the network by accessing a set of $k-2b$ genuine nodes in the genie-aided setting. These nodes are indexed as $\ell_{1}, \cdots, \ell_{k-2b}$. We assume that each one of the nodes $\ell_{1}, \cdots, \ell_{k-2b}$ is a replacement node, added to the network through a repair procedure. We also assume for any $i\in \{1,\cdots,k-2b\}$, the repair for node $\ell_{i}$ is performed after the repair for any node $\ell_{j},~j<i$, and all of the nodes $\ell_{j},~j\in\{1,\cdots,i-1\}$ are used as genuine helpers in the repair of the node $\ell_{i}$, as depicted in Fig. \ref{Fig_InformationFlowGraph}. However, note that we let the number of helpers participating in each repair to be independently chosen from the set $D$.
 
We now describe the procedure of forming a cut-set for the information flow graph described above, by choosing a subset of storage or repair edges. For any $i\in\{1,\cdots,k-2b\}$, we either choose its storage edge, or all the $d-2b-i+1$ repair edges coming from the genuine helpers not in the set $\{\ell_{1},\cdots,\ell_{i-1}\}$ to node ${\ell_{i}}_{\text{in}}$. In order to choose we compare $\alpha$, the capacity of the storage edge for node $\ell_{i}$, with the sum of the capacities of the described repair edges, and whichever is smaller its corresponding edges will be added to the cut-set. Then, considering all the possible choices for the number of helpers in all the $k-2b$ repair procedures, we can find the following upper-bound on the total storage capacity using the best cut-set achieved by this scheme as follows,
\begin{align}
F \leq \sum_{i=1}^{k-2 b}{\min\left(\alpha,\min_{d\in D}\left((d-2b-i+1)\beta(d)\right)\right)} = \sum_{j=0}^{k-2 b-1}{\min\left(\alpha,\min_{d\in D}\left((d-2b-j)\frac{\gamma(d)}{d}\right)\right)}. \nonumber
\end{align} 

In the case of the MBR mode however, we know from Theorem \ref{THM_MBR_OptimalRBW}, 
\begin{align}
\gamma_{\text{MBR}}(d) = \frac{\alpha d}{d-2b}. \nonumber
\end{align}
As a result, for any $j \in \{0, \cdots, k-2b-1\}$ we get
\begin{align}
\min_{d\in D}\left((d-2b-j)\frac{\gamma_{\text{MBR}}(d)}{d}\right) = (d_{\min}-2b-j)\frac{\alpha}{d_{\min}-2b}, \nonumber
\end{align}
and hence,
\begin{align}
F_{\text{MBR}} &\leq \sum_{j=0}^{k-2 b-1}{\min \left(\alpha,(d_{\min}-2b-j)\frac{\alpha}{d_{\min}-2b}\right)}. \nonumber \\
&= \frac{\alpha (k-2 b)}{d_{\min}-2 b}\left(d_{\min}-b-\frac{k-1}{2}\right). \nonumber
\end{align} 
\end{proof}

\bibliographystyle{IEEEtran}
\bibliography{IEEEabrv,DSS_Bibliography}

\end{document}